\documentclass[a4paper]{article}
\usepackage{amssymb}
\usepackage{graphicx}
\usepackage{epsfig}
\usepackage{amsmath,amsfonts,amssymb,amsthm}
\usepackage{mathtools}
\usepackage{wrapfig}
\usepackage{subfig}
\usepackage{caption}
\usepackage{complexity}
\usepackage[lined,boxed,commentsnumbered,ruled,vlined,noend,
linesnumbered,boxed]{algorithm2e}

\newcommand{\IR}{\mathbb{R}}

\newtheorem{theorem}{Theorem}
\newtheorem{lemma}{Lemma}

\newtheorem{definition}{Definition}
\newtheorem{observation}{Observation}
\newtheorem{invariant}{Invariant}

\usepackage[normalem]{ulem}
\usepackage{mdframed}
\usepackage{color}

\usepackage{hyperref}
\hypersetup{
	colorlinks,
	linktoc=all,
	pdfstartview={FitH}
}

\usepackage[usenames,dvipsnames]{xcolor}

\newcommand{\blue}[1]{{\textcolor{blue}{#1}}}

\newcommand{\bluec}[1]{{\color{blue!60!black}{\textbf{#1}}}}

\newcommand{\remove}[1]{}
\usepackage{textcomp}
\usepackage{lineno}
\newcommand{\QIT}{\textsf{\sc In-Place-Rectangular-Counting-Query-2d-Tree}}

\newcommand{\LPA}{\textsf{\sc LRR-Premitive-Algorithm-candidate-pair-$(p,q)$}}
\newcommand{\LRC}{\textsf{\sc LRC($P$)}}

\newcommand{\USt}{\textsf{\sc Update\_Stair}}
\newcommand{\Expl}{\textsf{\sc Explore}}
\newcommand{\Par}{\mathord{\emph{parent}}}
\newcommand{\LC}{\mathord{\emph{left-child}}}
\newcommand{\RC}{\mathord{\emph{ right-child}}}
\newcommand{\Current}{\mathord{\it Current}}
\newcommand{\State}{\mathord{\it state}}
\newcommand{\Level}{\mathord{\it level}}
\newcommand{\Count}{\mathord{\it Count}}

\title{Variations of largest rectangle recognition amidst a bichromatic point set\footnote{A preliminary version of this work 
titled ``Space-efficient Algorithms for Empty Space Recognition among
a Point Set in 2D and 3D" appeared in CCCG 2011.}}
\author{Ankush Acharyya$^1$ \and Minati De$^2$\thanks{Research supported by DST INSPIRE Faculty Grant (DST-IFA-14-ENG-75).} \and Subhas C. Nandy$^1$\and Supantha Pandit$^1$}
\date{$^1$Indian Statistical Institute, Kolkata 700108, India\\ $^2$Indian Institute of Science, Bangalore 560012, India}

\begin{document}

\maketitle

\begin{abstract}

Classical separability problem involving multi-color point sets is an 
important area of study in computational geometry. In this paper, we 
study different separability problems for bichromatic point set $P = 
P_r \cup P_b$ on a plane, where $P_r$ and $P_b$ represent the set of 
$n$ red points and $m$ blue points respectively, and the objective is 
to compute a monochromatic object of the desired type and of maximum size. 
We propose in-place algorithms for computing (i) an arbitrarily oriented
monochromatic rectangle of maximum size in $\IR^2$, and (ii) an 
axis-parallel monochromatic cuboid of maximum size in $\IR^3$. The time
complexities of the algorithms  for problems (i) and (ii) are 
$O(m(m+n)(m\sqrt{n}+m\log m+n \log n))$ and $O(m^3\sqrt{n}+m^2n\log n)$, 
respectively. As a prerequisite, we propose an in-place  construction of 
the classic  data structure the {\it k-d tree}, which was originally invented by J. 
L. Bentley in 1975. Our in-place variant of the $k$-d tree for a set of 
$n$ points in $\IR^k$ supports both orthogonal range reporting and counting 
query using $O(1)$ extra workspace, and these query time complexities are same 
as the classical complexities, i.e., $O(n^{1-1/k}+\mu)$ and $O(n^{1-1/k})$,
respectively, where $\mu$ is the output size of the reporting query. 
The construction time of this data structure is $O(n\log n)$. Both 
the construction and query algorithms are non-recursive in nature that 
do not need $O(\log n)$ size recursion stack compared to the previously 
known construction algorithm for in-place $k$-d tree and query in it.  
We believe that this result is of independent interest. We also propose 
an algorithm for the problem of computing an arbitrarily oriented rectangle 
of maximum weight among a point set $P=P_r \cup P_b$, where each point in 
$P_b$ (resp. $P_r$) is associated with a negative (resp. positive) real-valued
weight that runs in $O(m^2(n+m)\log(n+m))$ time using $O(n)$ extra 
space.
\end{abstract}
{\bf Keywords:}
Bichromatic point set; obstacle-free rectangle recognition; orthogonal range counting; 
in-place k-d tree; maximum weight rectangle recognition; space efficient algorithms.


\section{Introduction} \label{intro}
Given a bichromatic point set $P = P_r \cup P_b$,  where $P_r$ is the set of $n$ red points and $P_b$ is the set of $m$ blue
points, the basic separability problem is to find a separator $S$ such that the points in $P_r$ and
$P_b$ lie in two different sides of $S$ respectively. The motivation for studying this separability problem
for a bichromatic point set stems from its various applications in facility location, VLSI layout design, image analysis,
data mining, computer graphics and other classification based real life scenarios
\cite{cristianini2000introduction,dobkin1996computing,duda2012pattern,eckstein2002maximum,edmonds2003mining}. The bichromatic
separability problem also has its application to detect obstacle free separators. In the literature, different types of
separators like hyperplane \cite{megiddo1983linear}, circle \cite{o1986computing},
rectangle \cite{cortes2009bichromatic,eckstein2002maximum,van2009identifying},
square \cite{cabello2008covering,sheikhi2015separating} has been studied to optimize the objective function of the
corresponding problem. In this paper, we focus on designing  space-efficient algorithms for the following problems:

\vspace{-0.1in}
\begin{itemize}
 \item[\bluec{P1}] Computing an \blue{{\it arbitrarily oriented monochromatic rectangle of maximum size} ($LMR$)}  
 in $\IR^2$, where a rectangle $U$ is said to be monochromatic if it contains points of
 only one color in the proper interior of $U$, and the \blue{{\em size}} of $U$ is the number of points of that color
 inside or on the boundary of the rectangle $U$.

\vspace{-0.05in}
\item[\bluec{P2}] Computing an \blue{{\it arbitrarily oriented rectangle of maximum weight} $(LWR)$} in $\IR^2$, 
where each point in the set $P_b$ (resp $P_r$) is associated with negative (resp. positive)
real-valued weight, and the weight of a rectangle $U$ is the sum of weights of all the
points inside or on the boundary of $U$.

\vspace{-0.05in}
\item[\bluec{P3}] Computing a \blue{{\it monochromatic axis parallel cuboid\footnote{a solid which has
six rectangular faces at right angles to each other}} ($LMC$)} of maximum size in $\IR^3$.
\end{itemize}
\vspace{-0.1in}

A rectangle of arbitrary orientation in $\IR^2$ is called \blue{ red} if it does not contain 
any blue point in its interior{\footnote{blue points may appear on the boundary}}. The \blue{ 
largest red rectangle ($LRR$)} is a red rectangle of maximum size. Similarly,  the \blue{largest blue 
rectangle ($LBR$)} is defined. The \blue{ largest monochromatic rectangle ($LMR$)} is either 
$LRR$ or $LBR$ depending on which one is of maximum size. Here, the objective is to compute the 
$LRR$. In $\IR^3$,  we similarly define the \blue{ largest axis parallel red cuboid ($LRC$)}, 
i.e. a cuboid containing the maximum number of red points and no blue point in its interior$^1$. 
 We use $x(p)$ and $y(p)$ to denote the $x$- and $y$-coordinate of a point $p\in P$  respectively.

Several variations of this problem are well studied in the literature. In the well-known {\it 
maximum empty rectangle} (MER) problem, a set $P$ of $n$ points is given; the goal is to 
find a rectangle (axis parallel/arbitrary orientation) of maximum area that does not contain 
any point of $P$ in its interior (see \cite{aggarwal1987fast,chazelle1986computing,
naamad1984maximum,orlowski1990new} for MER of fixed orientation, and \cite{CND,asish} 
for MER of arbitrary orientation). For fixed orientation version, the best-known algorithm runs in 
$O(n\log^2 n)$ time and $O(n)$ space \cite{aggarwal1987fast}. For arbitrary 
orientations version, the best-known algorithm runs in $O(n^3)$ time using $O(n)$ space 
\cite{CND}. 

For the bichromatic version of the problem, Liu and Nediak \cite{liu2003planar} designed 
an algorithm for finding an axis parallel $LRR$ of maximum size in $O(n^2\log n+nm+m\log m)$
time using $O(n)$ space. Backer and Keil \cite{backer2009bichromatic} improved the time 
complexity to $O((n+m)\log^3 (n+m))$ using $O(n\log n)$ space adopting the divide-and-conquer 
approach of Aggarwal and Suri \cite{aggarwal1987fast}. They also proposed an $O((n+m)\log 
(m+n))$ time algorithm for finding an axis-parallel {\em red} square of maximum size. Recently, 
Bandyapadhyay and Banik \cite{bandyapadhyay2017polynomial} proposed an algorithm for finding the $LRR$ in arbitrary 
orientation  using $O(g(n,m)\log(n+m)+n^2)$ time and $O(n^2+m^2)$ space, where $g(n,m)\in 
O(m^2(n+m))$ and $g(n,m)\in \Omega(m(n+m))$.  

Other variations of the $LRR$ problem, studied in the literature are as follows. For a given 
bichromatic (red,blue) point set, Armaselu and Daescu \cite{Daescucccg16} considered the 
problem of finding a rectangle of maximum area containing all red points and minimum 
number of blue points. In $\IR^2$, the axis-parallel version of this problem can 
be solved in $O(m\log m +n)$ time and the arbitrary oriented version requires $O(m^3 
+ n\log n)$ time. In $\IR^3$, the axis-aligned version of the problem can be solved 
in $O(m^2(m+n))$ time. Eckstein {\it et al.} \cite{eckstein2002maximum} considered the 
axis-parallel version of the $LRR$ problem in higher ($d \geq 3$) dimensions. They showed that, 
if the dimension $d$ is not fixed, the problem is $NP$-hard. They presented an 
$O(n^{2d+1})$ time algorithm for any fixed dimension $d \geq 3$. Later, Backer and 
Keil \cite{backer2010mono} improved the time bound of the problem to 
$O(n^d\log^{d-2} n)$.  Cort{\'e}s {\it et al.} \cite{cortes2009bichromatic} considered
the problem of removing as few points as possible from the given bichromatic point 
set such that the remaining points can be enclosed by two axis-parallel rectangles 
$A_R$ and $A_B$ (may or may not be disjoint), where $A_R$ (resp. $A_B$) contains all 
the remaining red (resp. blue) points. They solved this problem in $O(n^2\log n)$ 
time using $O(n)$ space. The problem of separating bichromatic point sets by two 
disjoint axis-parallel rectangles such that each of the rectangles is monochromatic, 
is solved in $O(n\log n)$ time by Moslehi and Bagheri \cite{moslehi2016separating} (if 
such a solution exists). If these two rectangles are of arbitrary orientation then 
they solved the problem in $O(n^2\log n)$ time. Bitner {\it et al.} \cite{bitner2010minimum} 
studied the problem of computing the minimum separating circle, which is the 
smallest circle containing all the points of red color and as few points as possible 
of blue color in its interior. The proposed algorithm runs in $O(nm\log m+n\log n)$ 
time using $O(n+m)$ space. They also presented an algorithm for finding the largest 
separating circle in $O(nm\log m+k(n+m)\log (n+m))$ time using $O(n+m)$ space, where 
$k$ is the number of separating circles containing the smallest possible number of 
points from blue point set. The problem of covering a bichromatic point set with 
two disjoint monochromatic disks has been studied by Cabello {\it et al.} 
\cite{cabello2013covering}, where the goal is to enclose as much points as possible 
in each of the monochromatic disks. They solved the problem in 
$O(n^{\frac{11}{3}} \polylog~n)$ time. If the covering objects are unit disks or unit 
squares, then it can be solved in $O(n^{\frac{8}{3}}\log^2 n)$ 
and $O(n\log n)$ time respectively \cite{cabello2008covering}. 

The weighted bichromatic problems are also studied in the literature. The smallest 
maximum-weight circle for weighted points in the plane has been addressed 
by Bereg {\it et al.} \cite{bereg2015smallest}. For $m$ negative weight points and $n$ 
positive weight points they solved the problem in $O(n(n+m)\log (n+m))$ time 
using linear space. For a weighted point set Barbay {\it et al.} \cite{barbay2014maximum} 
provided an $O(n^2)$ time algorithms to find the maximum weight axis-parallel square.

\subsection*{Our Contribution}

Given a bichromatic point set $P=P_r \cup P_b$ in a rectangular region
$\cal A \subseteq$ $\IR^2$, where $P_r$ and $P_b$ are set of $n$ red points and $m$ blue 
points respectively, we design an in-place algorithm\footnote{An in-place algorithm
is an algorithm where the input is given in an array, the execution of the algorithm is
performed using only $O(1)$ extra workspace, and after the execution of the algorithm all
the input elements are present in the array.} for finding an arbitrarily oriented $LMR$ of
maximum size in $O(m(n + m)(m\sqrt{n} + m \log m + n\log n))$ time, using $O(1)$ extra workspace.
We also show that the axis-parallel version of the $LMR$ problem in $\IR^3$ 
(called the $LMC$ problem) can be solved in an in-place manner in $O(m^3\sqrt{n}+m^2n\log n)$
time using $O(1)$ extra workspace. As a prerequisite of the above problems, we propose an 
algorithm for constructing a k-d tree with a set of $n$ points in $\IR^k$ given 
in an array of size $n$ in an in-place manner such that the  orthogonal
range counting query can be performed using $O(1)$ extra workspace. The construction 
and query time of this data structure is $O(n\log n)$ and $O(n^{1-1/k})$, respectively.
Finally, we show that if the points in $P_r$ (resp. $P_b$) have positive (resp. 
negative) real-valued weight, then a rectangle of arbitrary orientation 
with maximum weight (called $LWR$) can be computed in $O(m^2(n+m)\log (n+m))$ time 
using $O(n)$ space.

\section{In-place k-d tree} \label{in-place}
To perform the orthogonal range reporting query, Bentley~\cite{Bentley75} invented
k-d tree   in 1975. It  is a binary tree in which every node is a k-dimensional 
point. Every non-leaf node can be thought of being associated with one of the 
k-dimensions of the corresponding point, with a hyperplane perpendicular to that 
dimension's axis, and implicitly this hyperplane  splits the space into two  
half-spaces. Points to the negative side of this \emph{splitting hyperplane} are 
represented by the left subtree of that node and points in the 
positive side of the hyperplane 
are represented by the right subtree\footnote{For a hyperplane $x_i=c$, its
{\it negative} (resp. {\it positive}) side is the
half-space $x_i<c$ (resp. $x_i > c$), where $x_i$ is the $i$-th coordinate 
of a $k$ dimensional point $(x_1, x_2, \ldots, x_k)$}. Depending on the level of a node 
going down the tree, the splitting dimension is chosen one after another in a
cyclic manner. Each node $v$ of the tree is associated implicitly with a 
rectangular region of space, called \emph{cell($v$)}. The cell 
corresponding to the root of the tree is the entire $\IR^k$. A child's cell is 
contained within its parent's cell, and it is determined by the splitting
hyperplane stored at the predecessor nodes.

Br\"{o}nnimann~{\it et al.}~\cite{BCC04}  mentioned an in-place version of
the k-d tree. We note that their approach for both constructing the data 
structure and querying in the data structure are recursive, and need to 
remember the subarray and the cell in which the recursive call is invoked. 
As a result, there is a hidden $O(\log n)$ space requirement for system 
stack. We present an alternate variant of in-place $k$-d tree data structure 
that supports both reporting and  counting query for orthogonal query
range with same query time as the classical one.  The advantage of this 
in-place variant is that both construction and query algorithms are 
non-recursive, and it takes only $O(1)$ extra workspace during the execution 
apart from the array containing the input points. The in-place organization 
of this data structure is similar to the in-place min-max priority search 
tree proposed by De {\it et al.}~\cite{de2013place}. 

\subsection{Construction of in-place k-d tree} \label{preprocessing2d}
Let us  consider that a set $P$ of $n$ points in $\IR^k$ is given in an array $P[1,\ldots,n]$.  
We propose an in-place algorithm to construct the k-d tree $\cal T$  in the array $P$.  Here, $\cal T$
is a binary tree of height $h = \lfloor \log n \rfloor$, such that the levels $0,1,\ldots,h-1$ are full
and level $h$ consists of $\varkappa=n - (2^h-1)$ nodes which are aligned as far as possible to the left.
At the end of the construction, the tree $\cal T$  is stored implicitly in the given array $P$. In other
words, we  store the  root of the  tree in  $P[1]$, its left and right children in $P[2]$ and $P[3]$, etc.
This allows us to  navigate  $\Par(P[i])$ which is at $P[\lfloor\frac{i}{2}\rfloor]$, and $\LC(P[i])$
and $\RC(P[i])$, if they exist, which are at $P[2i]$ and $P[2i+1]$, respectively.

\begin{figure}
 \noindent\begin{minipage}[b]{.42\textwidth}
 \centering{\includegraphics[width=2.2in]{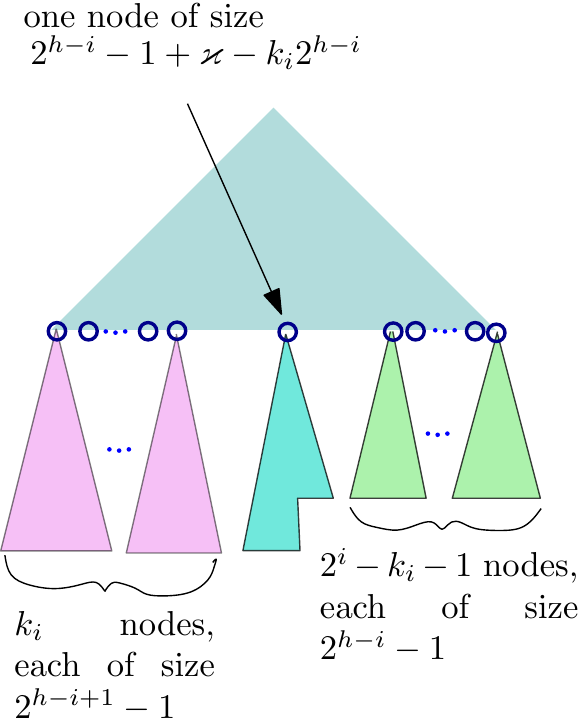}}
\centerline{$(a)$}
\end{minipage} 
\hfill
\begin{minipage}[b]{.58\textwidth}
\centering{\includegraphics[width=4in]{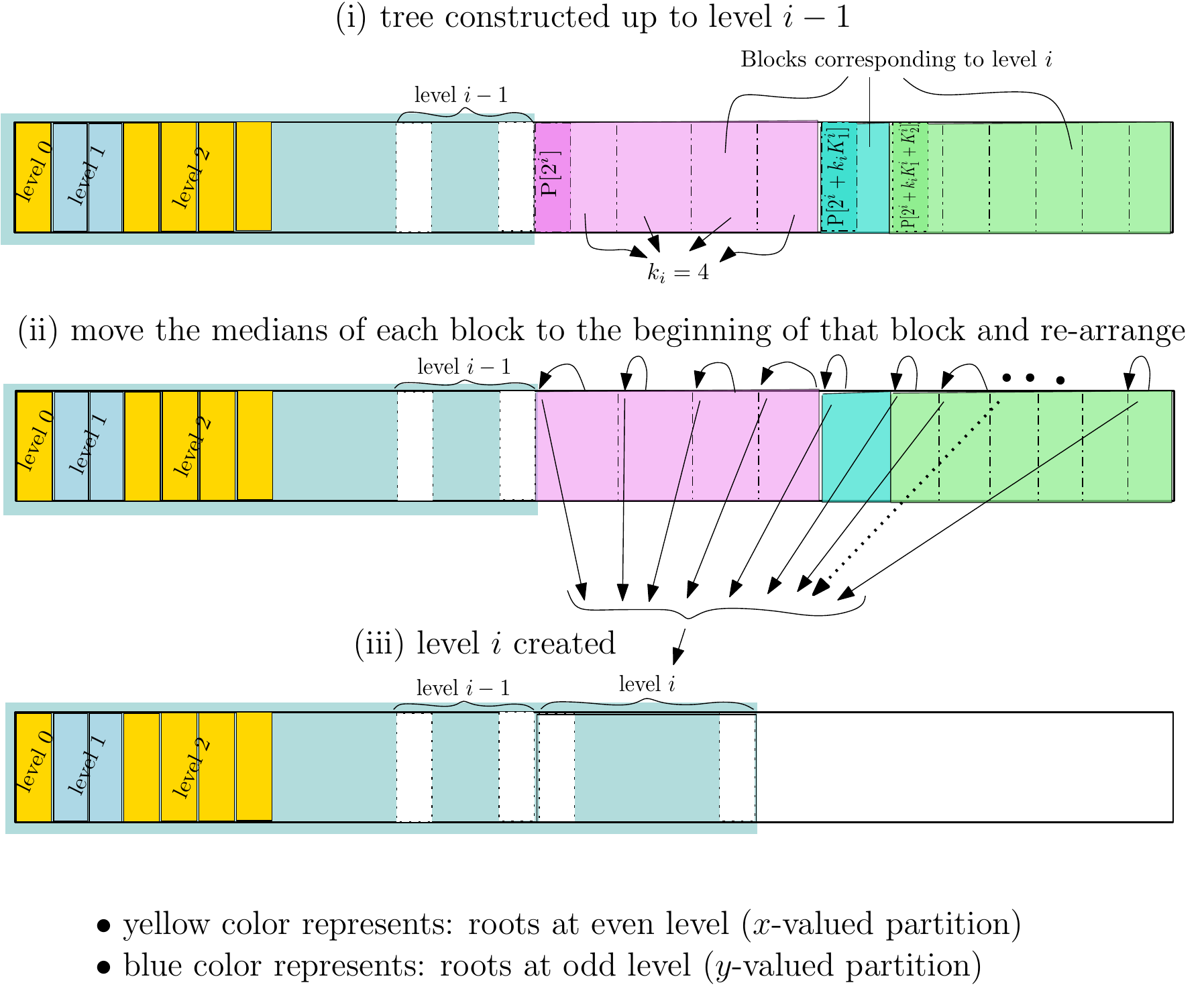}}
\centerline{$(b)$}
\end{minipage}
\caption{(a) k-d tree after constructing its $(i-1)$-th level (stripped), and (b) its array representation up to $i$-th level}\vspace{-0.15in} 
\label{2d-tree-1}
\vspace{-0.1in}
\end{figure}

Note that there are $2^i$  nodes in the level $i\neq h$  of the tree $\cal T$. As the number of leaf
nodes in a full tree of height $h-i$ is $2^{h-i}$,  so  there are $k_i = \lfloor{\frac{n-\varkappa}{2^{h-i}}}\rfloor$
nodes at level $i$ ($0<i<h$) that  are roots of subtrees, each of size $K_1^i= 2^{h+1-i} - 1$. If $k_i=2^i$, then
all the subtrees at level $i$ are full, and the number of nodes in each subtree is $2^{h+1-i}-1$. Otherwise,
we have $k_i< 2^i$, and level $i$ ($0<i<h$) consists of, from left to right,
\begin{itemize}
 \item $k_i$ nodes which are roots of subtrees, each of size $K_1^i= 2^{h+1-i} - 1$,
 \item one node that is the root of a subtree of size $K_2^i = 2^{h-i} - 1 + \varkappa - k_i\cdot2^{h-i}$, and
 \item  $2^i - k_i-1$ nodes which are roots of subtrees, each of size $K_3^i = 2^{h-i} - 1$.
\end{itemize}
See Figure \ref{2d-tree-1} for an illustration.

Here, we introduce the notion of {\em block} and {\em block median}.
Assume that  $0<i<h$. We refer to the portion of the array  $P[(2^i  +(j-1)K_1^i),\ldots, (2^i+jK_1^i-1)]$ as block $B_j^i$, for $j\leq k_i$.  The portion of the array   
$P[(2^i  +k_iK_1^i),\ldots, (2^i  +k_iK_1^i+K_2^i-1)]$ is  referred to as block $B_{k_i+1}^i$, and   $P[(2^i  +k_iK_1^i+K_2^i  + (j-1)K_3^i), \ldots 
(2^i  +k_iK_1^i+K_2^i  + jK_3^i-1)]$ are referred to as blocks $B_j^i$, for all $j>k_i+1$.
For $i=0$, we refer to the whole array $P[1,\ldots, n]$ as $B_1^0$. For $i=h$, we refer to the array element $P[2^h+j]$ as block $B_j^h$, where $1\leq j\leq \varkappa$.
For a block $B_j^i$ ($0<i<h$) of size $K_1^i$ (resp. $K_3^i$), we denote block median  $m_j^i$ as a point in $B_j^i$ whose $(i \mod k)$-th coordinate value is $\lceil\frac{K_1^i}{2}\rceil$-th (resp. $\lceil\frac{K_3^i}{2}\rceil$-th) smallest among all the points in  $B_j^i$.
If the size of  $B_j^i$ is $K_2^i$, then depending on  whether $K_2^i-(2^{h-i}-1)< 2^{h-i-1}-1$ or $K_2^i-(2^{h-i}-1)\geq 2^{h-i-1}-1$, we refer the
block median  $m_j^i$ as a point in $B_j^i$  whose $(i \mod k)$-th coordinate value is
 $K_2^i- (2^{h-i-1}-1)$-th  or $2^{h-i}$-th smallest among all the points in $B_j^i$.   For block $B_1^0$, depending on  whether $n-(2^{h}-1)< 2^{h-1}-1$ or 
 $n-(2^{h}-1)\geq 2^{h-1}-1$, we refer the block median  $m_1^0$ as a point in  $B_1^0$ whose  $1$st coordinate value is $n- (2^{h-1}-1)$-th  or $2^{h}$-th smallest  among all the points in $B_1^0$.

Our algorithm constructs the tree level by level. 
After constructing the
$(i-1)$-th level of the tree, it maintains the following invariants:
\begin{invariant}
\begin{itemize} 
\item[(i)] The subarray $P[1, \ldots, 2^i-1]$ stores levels $0,1,\ldots, i-1$ of the tree.
 \vspace{-0.05in}     
 \item[(ii)]  Block $B_j^i$ contains  all the elements of the $j$-th leftmost subtree of
 level $i$, for $j\in\{1,\ldots,2^i\}$ $(j\in \{1,\ldots, \varkappa\}$ when $i=h)$.   
\end{itemize}
\end{invariant}

At the first iteration of the algorithm,  we  find the block median $m_1^0$ using the linear time 
in-place median finding algorithm of Carlsson and Sundstr{\"o}m \cite{CarlssonS95}, and swap it with $P[1]$. Next,
we  arrange the subarray $P[2,\ldots, n]$ such that all the elements whose first coordinate value
is greater than $m_1^0$ appear before all the elements whose first coordinate value is
less than $m_1^0$.  We can do this arrangement in linear time using $O(1)$ extra space.

Note that after the first iteration, both the invariants are maintained.

Assuming that the tree is constructed up to level $(i-1)$, now, we construct the tree up to level $i$ by
doing the following:
\begin{enumerate}
\item  Find block median $m_j^i$ from each block $B_j^i$ and swap it with the first location of block
$B_j^i$.  Using the  median finding algorithm of
\cite{CarlssonS95}, this needs a total of $O(n)$ time for all the blocks in this ($i$-th) level.

\item Now depending on the median value $m_j^i$ we arrange  the elements of each block $B_j^i$ such
that all the elements having $(i \mod k)$-th coordinate value greater than $m_j^i$ appears before all 
the elements having  $(i \mod k)$-th coordinate value less than $m_j^i$. Thus the block $B^i_j$
splits into two parts, named {\it first half-block} and {\it second half-block}. This step again needs 
time proportional to the size of each block, and hence $O(n)$ time in total.

\item Now, we need to move all $m_j^i$  stored at the first position of each block to the correct
position of level $i$ of the tree. To do this we do the following.  First, we  move the last block
median $m_{2^i}^i$  next to $m_{2^i-1}^i$ by  two swaps; (i) swap $m_{2^i}^i$ with the first element of the 
second half-block of $B_{2^i-1}^i$, and (ii) swap $m_{2^i}^i$ with the first element of the first half-block of 
$B_{2^i-1}^i$. Thus, after this swapping step all the 
elements in the block $B_{2^i-1}^i$ that are less than $m_{2^i-1}^i$ will stay before the elements greater 
than $m_{2^i-1}^i$. Now, we will move both the pair ($m_{2^i-1}^i$, $m_{2^i}^i$) just after $m_{2^i-2}^i$. 
It can be shown that, for the move of each element of this pair, we need a pair of swaps as explained above. 
Next, we will move  $m_{2^i-2}^i$, $m_{2^i-1}^i$ and  $m_{2^i}^i$ by 
swapping (as mentioned above) next to $m_{2^i-3}^i$. In this way, we will continue until all 
the block medians $\{m_j^i| j\in 2^i\}$ will become consecutively placed. Using $O(1)$ space, this can be done in 
linear time\footnote{The reason is that, during this step of execution each element is moved backward 
from its present position in the array at most once.}.
\end{enumerate}

Step 3 ensures that both the invariants are maintained after this iteration.

At the end of  $h$-th iteration, we have the tree $\cal T$  stored implicitly in the array $P$.
The correctness of this algorithm follows by observing that the invariants are correctly maintained.   
As there are $O(\log n)$ iterations and each iteration takes $O(n)$ time, in total the algorithm 
takes $O(n \log n)$ time.

\begin{lemma} \label{preprocess}
Given a set of $n$ points in $\IR^k$ in an array $P$, the in-place construction of $k\text{-}d\text{-}tree$
takes $O(n\log n)$ time and $O(1)$ extra workspace. 
\end{lemma}

\subsection{Orthogonal range counting query in the in-place k-d tree}\label{counting2d}
For the simplicity of explanation, we illustrate the range counting query for points in $\IR^2$. 
We can  easily generalize  it for points in $\IR^k$, for any fixed $k$. Given a rectangular range 
$Q = [\alpha, \beta] \times [\gamma, \delta]$ as a query, here, the objective is to return a count 
of the number of elements in $P$ that lie in the rectangular range $Q$.

 In the  traditional model, to answer counting query in $O(\sqrt{n})$ time each node in pre-proceesed k-d tree
 stores the subtree size.  For our case, we cannot afford  to store the subtree size along with each node of the
 in-place k-d tree. However, if we have the information of the level  $\ell$ of a node $P[t]$, then we can
 on-the-fly compute  the subtree size as follows. Note that $P[t]$ is $r=t-(2^{\ell}-1)$-th left most node at
 $\ell$-th level of the  tree $\cal T$. Depending on whether $r\leq k_{\ell}$, $r= k_{\ell}+1$ or $r \geq k_{\ell}+2$,
 the subtree size of the node corresponding to $P[t]$ is $K_1^{\ell}$, $K_2^{\ell}$ or $K_3^{\ell}$.
 We want to remind the reader that $k_{\ell} = \lfloor{\frac{n-\varkappa}{2^{h-\ell}}}\rfloor$,
 $K_1^{\ell}= 2^{h+1-\ell} - 1$,  $K_2^{\ell}= 2^{h-\ell} - 1 + \varkappa - k_{\ell}\cdot2^{h-\ell}$
 and $K_3^{\ell}= 2^{h-\ell} - 1$, where $\varkappa=n-(2^h-1)$. 

On the other hand, the traditional query algorithm is a recursive algorithm that starts from the root of the tree.
At a node $v$, (i) if $Q\cap cell(v)=\emptyset$, then it returns 0; (ii) else if $cell(v)\subseteq Q$, then it returns
the subtree size of $v$; (iii) otherwise, it recursively  counts in the two children of $v$ and returns by adding
 these  counts, accordingly. The main  issue in implementing  this  algorithm  in the in-place model is that it
 needs  $O(\log n)$ space for system stack to have the knowledge of  the corresponding $cell$ of a node.  To tackle
 this situation, we have a new geometric observation which leads to a non-recursive algorithm in the  in-place model.

At a node $v$, we can test  whether the cells corresponding  to both the children  are intersecting the query
region $Q$ or not, by checking whether the  splitting plane stored at $\Par(v)$ is intersecting the query region
$Q$ or not. If the splitting plane does not intersect, then  the one of the child's cell that has non-empty intersection
with  $Q$,  can be decided by checking in which side of the hyperplane the region $Q$  lies. This simple trick works
because when we are at a node $v$, we know that the cell corresponding to its parent has non-empty intersection with $Q$.
 The following observation plays a crucial role here.

 \begin{observation}\label{obj:1}
 If the left (resp. right, bottom, and top)  boundary of $cell(v)$  intersects the  query region $Q$, then the
 left (resp. right, bottom, and top)  boundary of $cell(v')$ corresponding to the left (resp. right, left, right) 
 child ($v'$) of node $v$ also intersects the  query region $Q$.
\end{observation}

To decide whether $cell(v)\subseteq Q$, we do the following. Throughout the query algorithm,  we keep a
four-tuple $(L,R,B,U)$ each being able to store one coordinate value of the given input points. Initially,
all of them are set to $NULL$. Throughout the query algorithm, this  four-tuple  maintains the following invariant:

\begin{invariant}
 When we are at a node $\Current$,   the  non-NULL or NULL value stored at $L$  (resp. $R$, $B$, and $U$)
 implies that the left (resp. right, bottom, and top)  boundary of the $cell(\Current)$  is intersecting or not intersecting
 the query region $Q$. More specifically, if the value stored at $L$ (resp. $R$, $B$, and $U$) is non-NULL\footnote{split-value 
 of some node of the ancestor of $\Current$}, then it represents the left
 (resp. right, bottom, and top) boundary of the cell corresponding to the lowest level ancestor $v$ 
 of the node $\Current$,  such that left  (resp. right, bottom, and top) boundary of $cell(v)$ intersects the query region $Q$.
\end{invariant}

At a node $v$,  if  all the entries in the four-tuple is non-NULL, then the $cell(v)\subseteq Q$.
We present our algorithm as a pseudocode in Algorithm~\ref{AlgoRCount}. This is similar to the algorithm \Expl\ in~\cite{de2013place}.
It uses two variables $\Current$ and $\State$ that 
satisfies the following:

\noindent  
\begin{itemize}
\item $\Current$ is a node in $\cal T$.   
\item $\State \in \{0,1,2\}$. 
\item If $\State = 0$, then {\it either} $cell(\Current) \subseteq Q$ {\it or} both the children 
of $\Current$ need to be processed to compute $cell(\Current) \cap Q$.
\item If $\State = 1$, then all elements of the set $Q \cap \left( \{\Current\} \cup 
{\cal T}_{{\LC(\Current)}} \right)$ have been counted, where $ {\cal T}_{{\LC(\Current)}}$ 
is the left subtree of $\Current$ in the tree $ {\cal T}$. 
\item If $\State = 2$, then all elements of the set $Q \cap  {\cal T}_{\Current}$ 
      have been counted. 
\end{itemize}

\begin{algorithm}[h]
\SetAlgoLined 
\scriptsize
\SetKwData{Counnt}{Count}
\SetKwComment{Comment}{\*}{}
\KwIn{The root $p$ of $\cal T$ and a rectangular query range $Q = [\alpha, \beta] \times [\gamma, \delta]$.} 
\KwOut{Count of all the points $q$ in $T$ that lies in $Q$.} 
$\Current = p$;    $\State = 0$; $\Count=0$; $level=0$; 4-Tuple=$(L,R,U,B)=(NULL,NULL,NULL,NULL)$\; 
    \While{$\Current \neq p$ or $\State \neq 2$}
          {\If{$\State = 0$}
              {

       			   \If{$L\neq NULL \bigwedge R\neq NULL \bigwedge U\neq NULL \bigwedge B\neq NULL $}
			{
			$\Count=\Count + SubtreeSize({\cal T}_{\Current})$ \;
            \If{($level \mod 2 =0$)} {$L=val(\Current)$\; 
                    \If{($val(\Current) = R$)} {$R=NULL$\;}}
			\If{($level \mod 2 =1$)} {$B=val(\Current)$\; 
                \If{($val(\Current) = T$)} {$T=NULL$\;}}
		 	\If{$\Current$ is the $\LC$ of its parent}{$\State = 1$;}
			 \Else{$\State = 2$;}
			 $\Current = \Par(\Current)$; $\Level=\Level-1$\;
			}
                \Else{ \If{($val(\Current)$ lies in $Q$)}{$\Count=\Count+1$\;} 
               \If{$(\Current$ has a left child $)\bigwedge$
                (full or a part of $Q$ is in the left/bottom half-space of the splitting hyperplane at $\Current$ )} 
                  { 
                  \If{the splitting hyperplane at $\Current$  intersects $Q$}
                  {{\bf if} ($level \mod 2 =0$ and $R = NULL$) {\bf then} $R=val(\Current)$\;
                  {\bf if} ($level \mod 2 =1$ and $T = NULL$) {\bf then} $T=val(\Current)$\;}
                  $\Current = {\LC(\Current)}$; $\Level=\Level+1$\;               
                   } 
               \Else{$\State = 1$;}
                  
              }
              }
           \Else{\If{$\State = 1$} 
                   { {\If{$(\Current$ has a right child) $ \bigwedge$
                 (full or part of $Q$ is in the right/top half-space of the splitting hyperplane at $\Current$ )}
                        {\If{the splitting hyperplane at $\Current$  intersects $Q$}
                  {{\bf if} ($level \mod 2 =0$ and $L = NULL$) {\bf then} $L=val(\Current)$\;
                  {\bf if} ($level \mod 2 =1$ and $B = NULL$) {\bf then} $B=val(\Current)$\;}
                        $\Current = {\RC(\Current)}$;  $\Level=\Level+1$\; 
                         $\State = 0$\;
                         }}
                         
                     \Else{$\State = 2$;  }
                    }
                 \Else{\Comment{// $\State = 2$ and $\Current \neq p$}
                       \If{($\Current$ is the $\LC$ of its parent) $\bigwedge$ (the splitting hyperplane at $\Current$  intersects $Q$)} 
                          {$\State = 1$;} 
                       {{\bf if} ($level \mod 2 =0$ and $L = NULL$) {\bf then} $L=val(\Current)$\;
                  {\bf if} ($level \mod 2 =1$ and $B = NULL$) {\bf then} $B=val(\Current)$\;}  
                          $\Current = \Par(\Current)$; $\Level=\Level-1$\;
                      }
                } 
          }
\normalsize
  
\caption{RangeCounting}
\label{AlgoRCount}
\end{algorithm}

Update of the four-tuple  $(L, R, T, B)$  is done as follows. While searching with the query
rectangle $Q$ and with $\State=0$, when $Q$ is split by the split-line of the node
and the search proceeds towards one subtree of that node, we store the split-value (corresponding to the
split-line) in the corresponding variable of the four-tuple provided
it is not set earlier (contains $NULL$ value). During the backtracking, i.e, when $\State=2$,  if the split-value
of the current node matches with the corresponding variable in the four-tuple, then the corresponding entity of the four-tuple  is set
to $NULL$. Now, if backtracking reaches from left, we set $\State=1$. Since the right child of the current node needs to be processed,
we set the corresponding entity of four-tuple with the split-value stored at that node.

The correctness of the algorithm follows from maintaining the invariants and Observation~\ref{obj:1}. 
In the worst case, we might  have visited all the nodes whose corresponding cells  overlap on the  
orthogonal query rectangle $Q$. As the number of cells stabbed by $Q$ can be shown to be 
$O(\sqrt{n})$~\cite{Berg2008}, we have the following result.

\begin{lemma}\label{count}
 Given the in-place $2$-d tree  maintained in the array $P$ of size $n$, the rectangular range
 counting query can be performed in $O(\sqrt{n})$ time using $O(1)$ extra workspace.
\end{lemma}
We can generalize, the above algorithm for points in  $\IR^k$. The only difference is that we need $2k$-tuple instead of four-tuple. 
Assuming $k$ is a fixed constant, we have the following:
\begin{lemma}\label{count2}
Given the in-place k-d tree  maintained in the array $P$ of size $n$, the orthogonal range
counting query can be performed in $O(n^{1-1/k})$ time using $O(1)$ extra workspace.
\end{lemma}

\section{$LMR$ problem in arbitrary orientation}\label{twod}

In this section, we describe the method of identifying an arbitrarily 
oriented red rectangle of largest size for a given bichromatic
point set $P=P_r\cup P_b$ in $\IR^2$. The $LRR$ problem was solved by Bandyapadhyay and Banik \cite{bandyapadhyay2017polynomial},
considering the blue points as obstacles, using the following observation:

\begin{observation} \label{obj1} \cite{bandyapadhyay2017polynomial}
At least one side of a $LRR$ must contain two points $p,q$ such that $p \in P_b$ and
$q\in P_r\cup P_b$, and other three sides either contain at least one point of $P_b$, 
or is open (unbounded) (see Figure \ref{obs2ex}). 
\end{observation}

\begin{figure}[h]
 \centering{\includegraphics[scale=0.5]{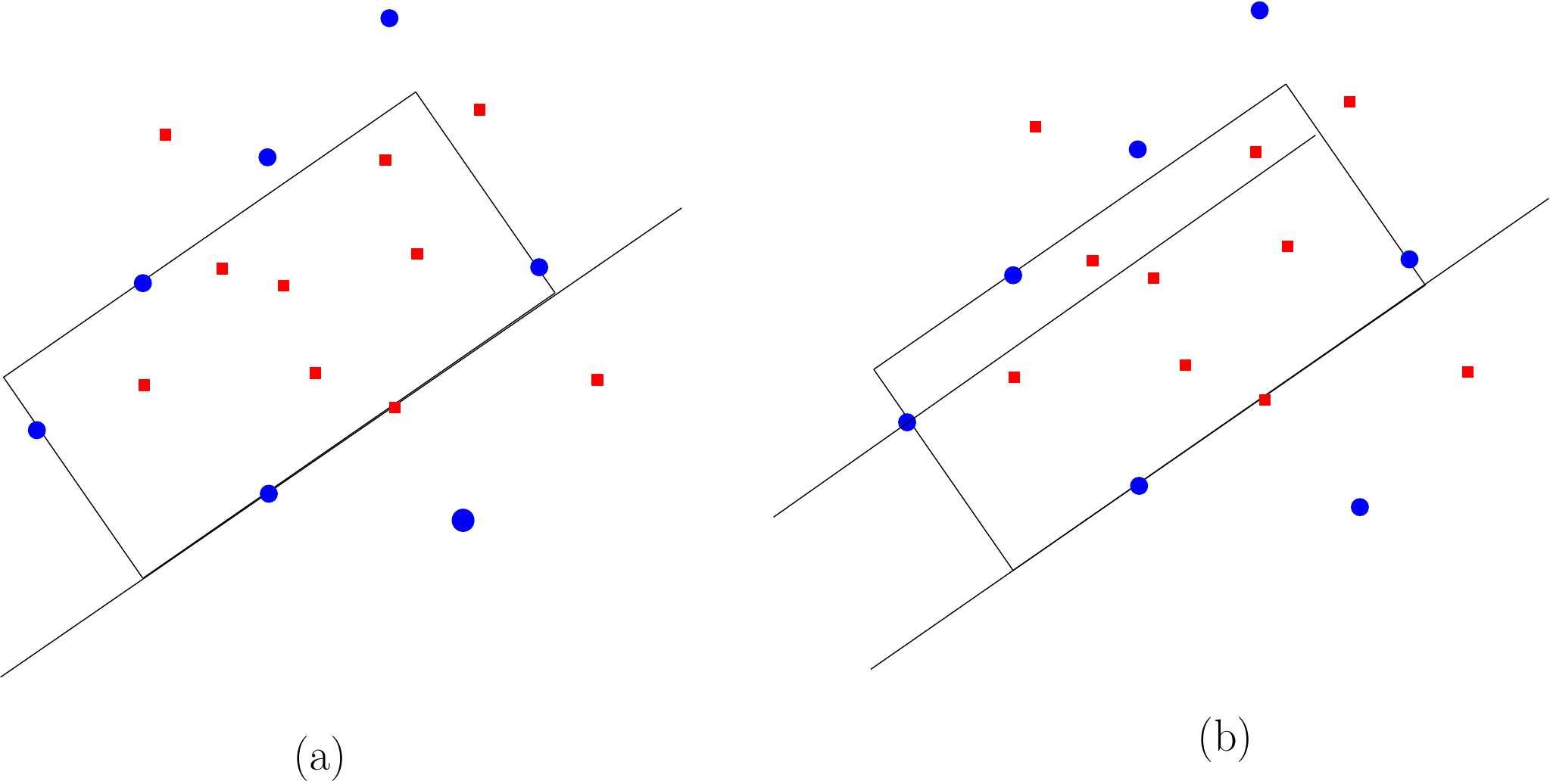}} 
 \caption{Example of $LRR$}
 \label{obs2ex}
\end{figure}

For the sake of formulation of our problem, let us have a {\it general position assumption} 
that no three points are collinear. We will use $\cal A$ to denote the convex hull of the point set $P$.
\begin{definition}
A pair of points $(p,q)$ is said to be a {\em candidate pair} if $p \in P_b$ and $q \in P_r\cup P_b$.
\end{definition}
\begin{definition}
A rectangle with one of its boundaries defined by a candidate pair, and each of the 
other three boundaries containing at least one point in $P_b$ is referred to as a {\em candidate $LRR$},
or $cLRR$ in short.
\end{definition}

We consider each candidate pair $(p,q)$, and define a line $\ell_{pq}$ passing through $p$ 
and $q$. We process each side of $\ell_{pq}$ separately to compute all the $cLRR$s' 
with $(p,q)$ on one of its boundaries by sweeping a line parallel to $\ell_{pq}$ 
among the points in $P$ in that side of $\ell_{pq}$, as stated below. After
considering all the candidate pairs in $P$, we report the $LRR$. We describe  
the method of processing the points in $P$ above\footnote{A point $(\alpha,\beta)$ 
is said to be above the line $ax+by+c = 0$ if $a\alpha + b\beta +c > 0$; otherwise 
the point $(\alpha,\beta)$ is below the said line.} $\ell_{pq}$. A similar method 
works for processing the points in $P$ below $\ell_{pq}$. 

\subsection{Processing a candidate pair $(p, q)$} \label{propq1}

Without loss of generality, we consider $\ell_{pq}$ as the $x$-axis, and $x(p) 
< x(q)$. Let $P'$ be the array containing the subset of $P$ lying above the $x$-axis. 
Let $P_b'$ and $P_r'$ denote the blue and red point set respectively in $P'$, $m'=
|P_b'|$ and $n'=|P_r'|$. We sort the points of $P_b'$ with respect to their 
$y$-coordinates, and construct a range tree $\cal T$ with the red points in $P_r'$ considering $\ell_{pq}$ as the $x$-axis.
 
Observe that each $cLRR$ above $\ell_{pq}$ with $(p,q)$ on its one side corresponds 
to a maximal empty rectangle ($MER$) \cite{cccgDeN11} among the points in $P_b'$ whose bottom 
side is aligned with the $x$-axis and containing $(p,q)$. We sweep a horizontal line 
$H$ in a bottom-up manner to identify all these $cLRR$s'.  

During the sweep, we maintain an interval ${\cal I}=[\alpha,\beta]$. $\cal I$ is initialized 
by $[x_{min},x_{max}]$ at the beginning of the sweep, where $x_{min}$ and $x_{max}$ 
are the points of intersection of the line $\ell_{pq}$ (the $x$-axis) with the 
boundary of $\cal A$. For each point $\theta \in P_b'$ 
encountered by the sweep line, if $x(\theta) \not\in {\cal I}$, sweep proceeds to process 
the next point.  Otherwise, we have a $cLRR$ with horizontal span $[\alpha,\beta]$, and the top 
boundary containing $\theta$\footnote{Needless to say, its bottom boundary contains 
the points $(p,q)$.}. Its size is determined in $O(\log n')$ time by performing a rectangular
range counting query in $\cal T$. Now, 
\vspace{-0.1in}
\begin {itemize}
 \item if $x(\theta) \in [x(p), x(q)]$ then the sweep stops.
 \vspace{-0.05in}
 \item otherwise, 
 \vspace{-0.1in}
 \begin{itemize}
\item if $\alpha \leq x(\theta) \leq x(p)$ then $\alpha =x(\theta)$ is set,
\vspace{-0.05in}
\item if $x(q) \leq x(\theta) \leq \beta$ then $\beta= x(\theta)$ is set, 
\end{itemize}
\end {itemize}
and the sweep continues. Finally, after considering all the points in $P_b$, the sweep
stops. For a detailed description of our proposed method, see Algorithm \ref{rangeT}.
A similar method is adopted for the points below $\ell_{pq}$.

\begin{algorithm}
 \SetAlgoLined
 \SetKwData{I}{I}\SetKwData{size}{size}\SetKwData{clrr}{cLRR}\SetKwData{Stop}{Stop}\SetKwData{sizec}{size(cLRR)}

\small 
 \KwIn{An array $P=P_b \cup P_r$ of points above $\ell_{pq}$; $P_b$ is $y$-sorted blue points, and 
$P_r$ corresponds to the range tree $\cal T$ for the red points.}
\tcc{$\ell_{pq}$ is the line through the candidate pair $(p,q)$}
\KwOut{$LRR$ in $P$}

$\alpha \leftarrow x_{min}$
\tcc*[r]{$x_{min}$ is left-intersection point of the line $\ell_{pq}$  with boundary of $\cal A$}
$\beta \leftarrow x_{max}$
\tcc*[r]{$x_{max}$ is right-intersection point of the line $\ell_{pq}$ with boundary of $\cal A$}
\I $\leftarrow [\alpha,\beta]$\;
\size $\leftarrow 0$
\tcc*[r]{number of red points in a rectangular range}
\sizec $\leftarrow 0$
\tcc*[r]{\size of optimum \clrr}

\For(\tcc*[f]{$H$ is the sweepline}){each point $\theta=(x_\theta,y_\theta) \in P_b$ encountered by $H$ in order} { 
\If{$x_\theta \in \text{\I}$}{
define a $cLRR$ with its bottom boundary by the candidate pair $(p,q)$, top boundary at 
$\theta$, left and right boundaries at $\alpha$ and $\beta$ respectively\;
determine \size of \clrr \tcc*[r]{using rectangular range query in $\cal T$}
\If{\size $>$ \sizec}{
\sizec $\leftarrow$ \size\;
}
\If{$x_\theta \in [x_p,x_q]$}{
\Stop\;
}
\If{$\alpha \leq x_\theta \leq x_p$}{
$\alpha \leftarrow x_\theta$\;}
\If{$x_q \leq x_\theta \leq \beta$}{
$\beta \leftarrow x_\theta$\;
}
}
}
return \sizec\;
\caption{\LPA} 
\label{rangeT}
\end{algorithm}

\begin{lemma} \label{l}
The above algorithm computes the $LRR$ in  $O(m(m+n)(m\log n + m\log m +n\log n))$
time using $O(n\log n)$ extra space. 
\end{lemma}
\begin{proof}
The space complexity follows from the space needed for maintaining the range tree $\cal T$ . We
now analyze the time complexity. For each candidate pair $(p,q)$, (i) the 
preprocessing steps sorting of the points in $P_b'$, and constructing $\cal T$ with the points in $P_r'$) need 
$O(n'\log n' + m'\log m')$ time, and (ii) during the sweep, reporting the size of each $cLRR$ needs 
$O(\log n')$ time \footnote{the time for the counting query for a rectangle in a range tree using 
fractional cascading.}. Since, $O(m')$ $cLRR$ may be reported for the candidate pair $(p, q)$, the 
total processing time for $(p, q)$ is $O(m'\log n' + m'\log m'+ n'\log n')$ in the worst case. The 
result follows from the fact that we have considered $O(m(n + m))$ candidate pairs, $m'= O(m)$ and 
$n'= O(n)$ in the worst case.
\end{proof}
The same method is followed to compute the $LBR$. Finally, $LMR$ is reported by comparing the size of 
$LRR$ and $LBR$. Lemma \ref{l} says that both the time and space complexities of our proposed algorithm 
for computing the $LMR$ are an improvement over those of the algorithm of 
Bandyapadhyay and Banik \cite{bandyapadhyay2017polynomial} for the same problem.  

It needs to be mentioned that, we can implement the algorithm for the $LRR$ problem in an in-place
manner by replacing range tree with the in-place implementation of 2-d tree as described in
Section \ref{in-place} for the range counting. Thus, the preprocessed data structure (the sorted
array of $P_b$ and the 2-d tree for $P_r$) can be stored in the input array $P$ without any extra
space. Using the results in Lemmata \ref{preprocess} and \ref{count}, we have the following result. 

\begin{theorem} \label{lrr}
In the in-place setup, one can compute an $LMR$ in $O(m(m+n)(m\sqrt{n}+m\log m+n\log n))$
time using $O(1)$ extra space. 
\end{theorem}

\section{$LWR$ problem in arbitrary orientation}

In this section, we consider a weighted variation of {\bf P1}. Here each point in $P_r$ 
is associated with a non-zero positive weight and each point in $P_b$ is associated with a 
non-zero negative weight. Our goal is to report a rectangle $LWR$ of arbitrary orientation 
such that the sum of weights of the points inside that rectangle (including its boundary) is 
maximum among all possible rectangles in that region. Unlike problem {\bf P1}, here the 
optimum rectangle may contain points of both the colors.

\begin{observation}
At least one side of the $LWR$ must contain two points $p,q\in P_r$, and other three sides 
either contain a point of $P_r$ or is open. A point $p\in P_r$
may appear at a corner of the solution rectangle $LWR$. In that case, $p$ is considered to 
be present in both the adjacent sides of $LWR$.
\end{observation}

We will consider all possible pairs of points $(p,q)\in P_r$ and define a line $\ell_{pq}$ 
joining $p,q$. We process each side of $\ell_{pq}$ separately to compute all the candidate 
$LWR$, denoted as $cLWR$, among the points in $P$ lying in that side of $\ell_{pq}$. After 
considering all possible pairs of points, we report $LWR$. We now describe the processing
of the set of points $P' \in P$ that lies above $\ell_{pq}$.

\subsection{Processing a point-pair $(p, q)$}
As earlier, assume $\ell_{pq}$ to be the $x$-axis. Consider a rectangle $R$ 
whose bottom side aligned with $\ell_{pq}$ (see Figure \ref{fig4}); the 
top side passing through $p_\theta$, left and right sides are passing through 
$p_b$ and $p_c$ respectively. We can measure the weight of the rectangle $R$ 
as follows:
 
 \begin{figure}[h]
\centerline{\includegraphics[scale=0.5]{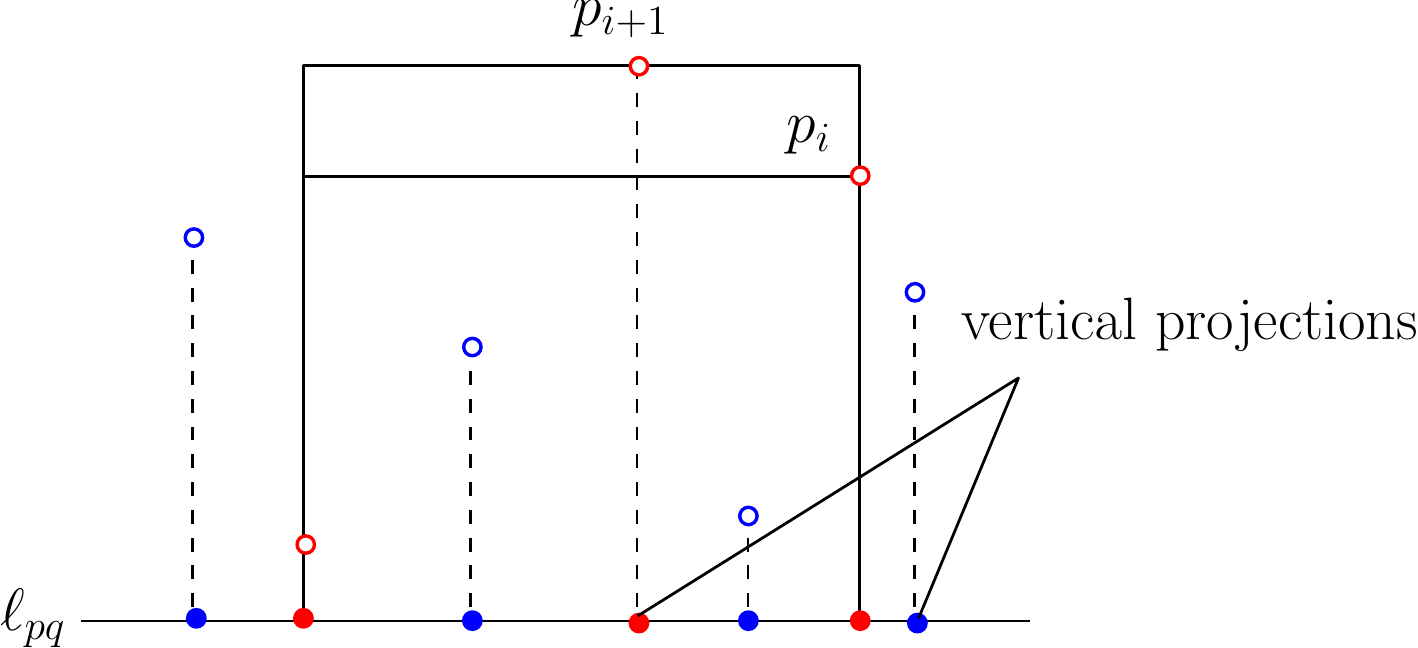}}
\caption{Update of $LWR$}
\label{fig4}
\end{figure}
 
 \begin{description}
  \item Let $U = \{u_i, i = 1, 2, \ldots,n\}$ be the projection of all the
   points on $\ell_{pq}$ having $y$-coordinate (distance from $\ell_{pq}$) 
   less than or equal to that of $p_\theta$. Each member $u_i$ is assigned 
   an weight equal to the weight of its corresponding point $p_i$. Now, 
   compute the cumulative sum of weights $W(u_i)$ at each projected point 
   of $U$ from left to right. Observe that the weight of the rectangle $R$ 
   is equal to $W(c)-W(\alpha)$, where
    $u_\alpha$ is the rightmost point in $U$ to the left of $p_b$.
 \end{description}

Thus, in order to get a maximum weight rectangle with its top boundary 
passing through the point $p_\theta$ and having $(p, q)$ on its bottom 
boundary, we need to search for an element in $u_{\alpha}\in U$ having 
$x$-coordinate less that $\min(x(p_\theta),x(p))$ having minimum weight, 
and an element $u_{\beta}\in U$ having $x$-coordinate greater than 
$\max(x(p_\theta),x(q))$ having maximum weight. \blue{The weight of the 
rectangle with $(p,q),p_\alpha, p_\theta, p_\beta$ on its bottom, left, 
top and right boundaries will be $W(u_{\beta})-W(u_{\alpha})$.} 

We sweep a horizontal line (see Figure \ref{fig4}) among the points 
in $P'$. During the sweep, we create a projection $u_i$ of each point 
$p_i\in P'$ and assign its weight $w(u_i) = \displaystyle w(p_i)$, 
and store them in a dynamically maintained weight balanced leaf search 
binary tree $\cal T$ \cite{overmars1983}. Its leaves correspond to the 
projections of all points that are faced by the sweep line (see Figure 
\ref{fig5}). Each internal node $u$ in $\cal T$ maintains three pieces 
of information, namely $EXCESS$, $MAX$ and $MIN$. $MAX$ and $MIN$ store 
the maximum and minimum of $W(u_i)$ values stored in the subtree rooted 
at the node $u$ of $\cal T$. The $EXCESS$ field is initialized with 
``zero''. Each projected point $u_j$ at the leaf also stores the 
cumulative sum of weights $W(u_j)$. During the sweep, when a new point 
$p_i \in P'$  is faced by the sweep line, $u_i$ is inserted in $\cal T$. 
Now, for all $u_j$ with $x(u_j) > x(u_i)$, the cumulative sum of weights 
needs to be updated as $\hat{W}(u_j)=W(u_j)+w(u_i)$. We use $EXCESS$ 
field to defer this update as follows. 

\begin{description}
 \item While tracing the search path to insert $u_i$ (= $x(p_i)$) in $\cal T$, 
 if the search goes from a node $v$ to 
 its left child, then we add $w(u_i)$ with the $EXCESS$ field of 
 the right child $z$ of $v$.
This is in anticipation that while processing another point $u_j \in P'$ if the 
search goes through $z$, then the $EXCESS$ field of $z$ will be 
propagated to the $EXCESS$ field of 
its two children (setting the $EXCESS$ field of $z$ to $0$).
\end{description}

\begin{figure}[h]
\centering
\includegraphics[scale=0.55]{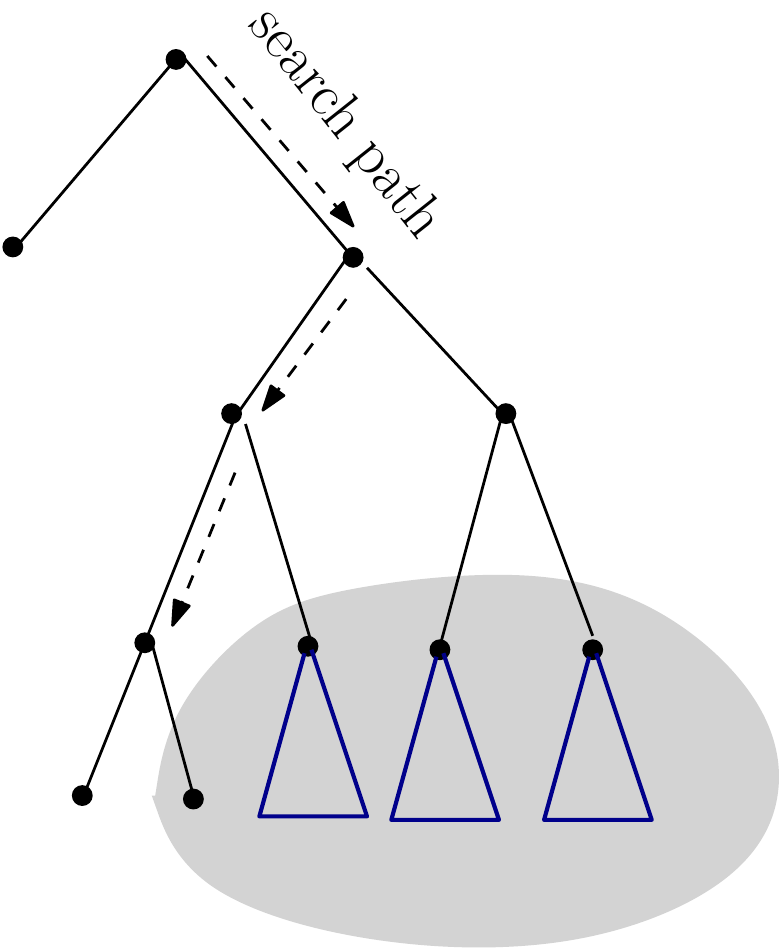}
\caption{Search Path in $\cal T$}
\label{fig5}
\end{figure}

After the insertion of $u_i$ in $\cal T$, we trace back up to the root of 
$\cal T$ and update the $MAX$ and $MIN$ fields (if necessary) of each node 
on the search path. If the (weight-)balance condition at any node is 
violated, a time linear in the size of the subtree rooted at that node is 
spent to rebuild that subtree in the (weight-)balanced manner. 

Now, if $p_i \in P_r$, then we find the $cLWR$ of maximum 
weight with $(p,q)$ on its bottom boundary and $p_i$ on its top boundary 
by identifying (i) a element $u_\alpha \in {\cal T}$ with $W(u_\alpha) 
= min\{W(u)|x(u) < min(x(p), x(p_i))\}$ using the $MIN$ fields of the 
nodes on the search path, and (ii) a point $u_\beta\in {\cal T}$ with 
$W(u_\beta) = max\{W(u)|x(u) > max(x(q),x(p_i))\}$ using the $MAX$ fields 
of the nodes on the search path. As mentioned earlier, the weight of the 
rectangle on $\ell_{pq}$ with $p_\alpha,p_i,p_\beta \in P_r$ on its left, 
top, and right sides respectively,  is $W(u_\beta)- W(u_\alpha)$. The 
iteration continues until all the points of $P'$ are considered by the 
sweep line.

\begin{lemma} \label{fixed_side}
The $cLWR$ of maximum weight with $(p,q)$ on its one side can be computed in $O((n+m)\log(n+m))$ time.
\end{lemma}
\begin{proof} 
Follows from the fact that the amortized insertion time of a point in 
$\cal T$ is $O(\log n)$ \cite{overmars1983}. While rebuilding, due to 
the violation of balance condition, the setting of $EXCESS$, $MIN$ and 
$MAX$ fields of each node can also be done in $O(|{\cal T }|)$ time, 
and rebuilding of $\cal T$ is needed after at least $O(\log n)$ updates 
\cite{overmars1983}.
\end{proof}
The algorithm proposed above is not in-place. It uses a preprocessed data 
structure implemented in an $O(n)$ extra space. Lemma \ref{fixed_side} and 
the fact that we need to consider $\ell_{pq}$ for each pair $p,q \in P_r$ 
suggest the following result:
 
\begin{theorem}
An $LWR$ of arbitrary orientation for a set of weighted  points 
can be computed in $O(m^2(n + m) \log(n + m))$ time using $O(n)$ workspace.
\end{theorem}

\section{Computing largest axis-parallel monochromatic cuboid 
$\IR^3$} \label{3d}
We now propose an in-place algorithm for computing a monochromatic 
axis-parallel cuboid with the maximum number of points. Here, the input 
is a set of bi-chromatic points $P=P_r\cup P_b$ inside a $3D$ 
axis-parallel region $\cal A$ bounded by six axis-parallel planes, 
where $P_r$ is the set of $n$ {\em red} points and $P_b$ is the set 
of $m$ {\em blue} points. The input points are given in an array, 
also called $P$. The $x,y,z$ coordinates of a point $p_i\in P$ are 
denoted by $x(p_i)$, $y(p_i)$ and $z(p_i)$ respectively, along with 
its color information $c(p_i)$ = red/blue. A cuboid is said to be a 
candidate for $LRC$ if its every face either coincides with a face of 
$\cal A$ or passes through a blue point, and its interior does not 
contain any blue point. Such a cuboid will be referred to as $cLRC$. 
The objective is to identify an $LRC$, which is a $cLRC$ containing the 
maximum number of red points. Similarly, a blue cuboid containing the 
maximum number of blue points ($LBC$) can be defined. The $LMC$ is 
either $LRC$ or $LBC$ depending on whose number of points is more.

We compute all possible maximal empty cuboid \cite{NB} among the $m$ blue 
points. Each one will be a $cLRC$; we perform a range query to count the 
number of red points it contains. In our algorithm, three types of $cLRC$s' 
inside $\cal A$ are considered separately. 
\begin{description} 
\item[type-1:] the $cLRC$ with both top and bottom faces aligned with the top 
and bottom faces of $\cal A$, 
\item[type-2:] the $cLRC$ whose top face is aligned with the top face of 
$\cal A$, but bottom face passes through a blue point in $P_b$, and
\item[type-3:] the $cLRC$ whose top face passes through some blue point 
in $P_b$. The bottom face may pass through another blue point in $P_b$ or 
may coincide with the bottom face of $\cal A$.
\end{description}

As a preprocessing, we first split the array $P$ into two parts, namely 
$P_r$ and $P_b$, such that $P_r = P[1,\ldots,n]$ and $P_b=P[n+1,\ldots,n+m]$.  
We construct an in-place 2-d tree $\cal T$ with the points in $P_r$ 
considering their $(x,y)$ coordinates, which will be used for the 
range-counting query for the $cLRC$s'. We also sort the points in $P_b$ 
in decreasing order of their $z$-coordinates. Thus, the preprocessing 
needs $O(m\log m + n\log n)$ time. 

In \cite{DumitrescuJ13,kaplan2008}, it is proved that the number of maximal 
empty hyper-rectangles among a set of $n$ points in $\IR^d$ is $O(n^d)$. 
In the following subsections, we will analyze the processing of these 
three types of $cLRC$s' in an in-place manner. The largest among the 
{\em type-$i$} $cLRC$ will be referred to as {\em type-$i$} $LRC$, for 
$i=1,2,3$. 

\remove{
\begin{algorithm}
\KwIn{An array containing the point set $P=\{p_1, p_2,\ldots,p_{m+n}\}$ in 
3-d axis-parallel box $\cal A$, where $p_i=(x_i,y_i,z_i)$.}
\KwOut{The largest axis-parallel red cuboid $(LRC)$ $C$, and its size 
$size_{max}$.} 
$size_{max}=0$\;
Execute TYPE-1\_LRC($size_{max},C$)\; 
Execute TYPE-2\_LRC($size_{max},C$)\;
Execute TYPE-3\_LRC($size_{max},C$)\;
Report $size_{max}, C$\;
\caption{\LRC}
\label{LRC}
\end{algorithm}
}

\subsection{Computation of {\it type-1} $LRC$} \label{type1}

As both the top and bottom faces of the {\em type-1} $cLRC$s' are aligned 
with the top and bottom faces of $\cal A$, if we consider the projections 
of the points in $P_b$ on the top face of $\cal A$, then each maximal 
empty axis-parallel rectangle ($MER$) on the top face of $\cal A$ will 
correspond to a {\it type-1} $cLRC$. Thus, the problem reduces to 
the problem of computing all the $MER$s' using the array $P_b$ in an in-place manner, and 
for each $MER$, count the number of points of $P_r$ inside the corresponding 
{\em type-1} $cLRC$ using the 2-d tree $\cal T$ with the projection of 
points in $P_r$ on the top face of $\cal A$. 

\remove{
\begin{algorithm}
 \SetKwData{return}{return}
 \KwIn{$P_r$ and $P_b$}
 \KwOut{TYPE-1 $LRC$}
 Consider the projections $b_i,~ \forall ~p_i \in P_b$ and $r_i, ~ \forall 
 ~p_i \in P_r$ on the top face of $\cal A$\;
 Construct a 2d-Tree $\cal T$ with the red points $r_i, i= 1\ldots,n$ 
 \tcc*{see Section \ref{preprocessing2d}}
 Compute all $MER$s for the projected points $b_i, i= 1,\ldots,m$ 
 \tcc*{see Section \ref{twod}}
 \For{ each generated MER}{perform range counting in $\cal T$ 
 \tcc*{see Section \ref{counting2d}}} 
 \return cuboid $C$ with maximum size $size_{max}$\;
 \caption{TYPE-1\_LRC($size_{max},C$)}
 \label{type-1}
\end{algorithm}
}

\begin{lemma}\label{res-type1}
The number of {\em type-1} $cLRC$ is $O(m^2)$ in the worst case and the one of maximum size can be computed in $O(m^2\sqrt{n}+n\log n)$ time.
\end{lemma}
\begin{proof}
The first part of the result i.e the number of {\em type-1} $cLRC$ follows from \cite{naamad1984maximum}. (i) We can generate all the 
$MER$s with bottom boundary passing through a point $b_i$
on the top face of $\cal A$ using the method described in Section \ref{twod} in $O(m)$ time, and (ii) 
for each $MER$, the number of projected red points inside that $MER$ can be obtained in $O(\sqrt{n})$ time using the 2-d tree $\cal T$.  
 The second part in the result follows from the fact that $\cal T$ can be generated in $O(n\log n)$ time (see Section \ref{preprocessing2d}).
\end{proof}

\subsection{Computation of {\it type-2} $LRC$} \label{type-2}
Now we describe the in-place method of computing the largest {\em type-2} $cLRC$ whose top face is aligned
with the top face of $\cal A$, but bottom face passes through a point in $P_b$. 
 
We will use $p_1,p_2, \ldots, p_m$ to denote the points in $P_b$ in decreasing order of their $z$-coordinates. 
We consider each point $p_i \in P_b$ in order and compute $LRC(p_i)$, the largest {\it type-2} red cuboid 
whose bottom face passes through $p_i$. Let $B_i=\{b_1, b_2, \ldots, b_{i-1}\}$, $i<m$ be the set containing the projection of all 
the blue points having $z$-coordinate larger than $z(p_i)$ on the plane $H(p_i)$. Similarly, $R_i=\{r_1,r_2, \ldots\}$ are 
the projection of all the red points having $z$-coordinate larger than $z(p_i)$ on the plane $H(p_i)$.
Thus, $LRC(p_i)$ corresponds to a rectangle on the plane $H(p_i)$ that 
contains $p_i$, but no point of $B_i$ in its interior, and has the maximum number 
of points of $R_i$. 

As in the earlier section, we can partition the array $P$ into two contiguous 
blocks $P_b$ and $P_r$. The block $P_b$ contains all the blue points in 
decreasing order of their $z$-coordinates. The block $P_r$ contains all the 
red points. The blue points are processed in top-to-bottom order. Global 
counters $LRC$ and $MAX_r$ are maintained to store the $LRC$ detected so 
far, and its size. While processing each point $p_i \in P_b$, let $B_i$ 
denote the blue points with their $z$-coordinates greater than $z(p_i)$. 
We split $P_r$ into two parts. The left part contains an in-place 2-d tree 
${\cal T}_i$ with all the red points having $z$-coordinates greater than 
$z(p_i)$. The right part of $P_r$ contains the red points with $z$-coordinates 
less than $z(p_i)$. We can compute all the $MER$s using the set $B_i$ as 
in Section \ref{type1}. For each generated $MER$ if it contains $p_i$ in 
its interior, then we perform the range counting query in ${\cal T}_i$ to 
compute the number of red points inside it. $LRC$ and $MAX_b$ are updated, 
if necessary. Thus, we have the following result.  

\begin{lemma} \label{res-type2}
The {\em type-2} $LRC$ can be computed in $O(m^3\sqrt{n}+mn\log n)$ time.  
\end{lemma}
\begin{proof}
The time complexity of processing each point $p_i \in P_b$ 
follows from Lemma \ref{res-type1}. Since $m$ blue points are processed, 
the result follows. 
\end{proof}

\subsection{Computation of {\it type-3} $LRC$} \label{type-3}
Here also we use $p_1,\ldots,p_m$ to denote the points in $P_b$ in decreasing 
order of their $z$-coordinates, and the algorithm processes the members in 
$P_b$ in this order. We now describe the {\it phase of processing} of a point 
$p_i \in P_b$. It involves generating all the {\it type-3} $cLRC$s whose top 
face passes through $p_i$; their bottom face may pass through another blue 
point $p_j\in P_b$ or may coincide with the bottom face of $\cal A$. 
Consider the horizontal plane $H(p_i)$ passing through $p_i\in P_b$ and sweep 
it downwards until it hits the bottom face of ${\cal A}$. During this phase when the sweeping 
plane touches $H(p_j)$ (i.e. hits a point $p_j \in P_b$), the points inside 
these two horizontal planes $H(p_i)$ and $H(p_j)$ will participate in computing 
the $cLRC$s with top and bottom faces passing through $p_i$ and $p_j$, 
respectively. Let, $B_{ij}= \{b_i,\ldots,b_j\}$ be the projections of 
these blue points $p_i,\ldots,p_j$ $(1 \leq i <j\leq n)$ on the plane $H(p_i)$.
Similarly, consider the projections $R_{ij}$ of the red points 
on the plane $H(p_i)$ those lie in between the planes $H(p_i)$ and $H(p_j)$. 
Our objective is to determine a $cLRC$ corresponding to an $MER$ on the plane 
$H(p_i)$ with the points in $B_{ij}$ as obstacles that contains the maximum number 
of points in $R_{ij}$.

In the phase of processing $p_i \in P_b$, the points of $P$ above $H(p_i)$ 
does not participate in this processing. Those points of $P_b$ (resp. $P_r$) 
are separately stored at the beginning of the array $P_b$ (resp. $P_r$). 
From now onwards, by $P_b$ (resp. $P_r$) we will mean the blue (resp.
red) points below $H(p_i)$. 

We consider two mutually orthogonal axis-parallel lines $x=x(p_i)$ and 
$y=y(p_i)$ on the plane $H(p_i)$ that partition $H(p_i)$ into four quadrants. 
The blue points that belong to the $\theta$-th quadrant, are denoted by 
$P_b^\theta$, and are stored consecutively in the array $P_b[i+1, \ldots, m]$. 
We use $m_\theta=|P_b^\theta|$. While processing the point $p_j \in P_b$ 
during the sweep in this phase, we use $B_{ij}^\theta$ to denote  
the projections of the subset of points in $P_b^\theta$ that lie between the planes 
$H(p_i)$ and $H(p_j)$, $\theta=1,2,3,4$. The members in $B_{ij}^\theta$ are 
stored in the consecutive locations of the array $P_b^\theta$ in decreasing 
order of their $z$-coordinates. We maintain four index variables $\chi_\theta$, 
$\theta=1,2,3,4$, where $\chi_\theta$ indicates the last point hit by the 
sweeping plane in the $\theta$-th quadrant. Thus, $p_j \in P_b\setminus 
(\cup_{\theta=1}^4 B_{ij}^\theta)$, and is obtained by comparing the 
$z$-coordinates of the points $\{P_b[\chi_\theta+1], \theta=1,2,3,4\}$.
We will use $R_{ij}$ to denote the projection of the points in $P_r$ 
lying between $H(p_i)$ and $H(p_j)$. These are stored at the beginning of the 
array $P_r$.

In each quadrant $\theta$, we define the unique maximal closest stair 
$\textit{STAIR}_\theta$ around $p_i$ with a subset of points of 
$B_{ij}^\theta$ as in \cite{cccgDeN11,NB}. The projection 
points of $B_{ij}^\theta$, that determine $\textit{STAIR}_\theta$, are stored
at the beginning of the sub-array $P_b^\theta$ in order of their $y$-
coordinates\footnote{The remaining elements ($B_{ij}^\theta \setminus 
\textit{STAIR}_\theta$ are stored just after $\textit{STAIR}_\theta$ in a
contiguous manner in $P_b^\theta$ so that the first unprocessed element 
in the quadrant $\theta$ is obtained at $P_b[\chi_\theta+1]$.}. Thus, $\bigcup_{\theta=1}^4 
\textit{STAIR}_\theta$ forms an empty ortho-convex polygon $OP$ on $H(p_i)$ 
(see Figure \ref{3dtype2i}(a)). As a consequence, the problem of finding a 
{\it type-3} $LRC$, with top and bottom faces passing through $p_i$ and $p_j$ 
respectively, maps to finding an $MER$ inside this ortho-convex polygon 
that contains $b_j$ and maximum number of points in the set $R_{ij}$.  
\remove{The 
pseudo-code of computing Algorithm $\textit{STAIR}_1$ is given in 
Algorithm \ref{stair1}. $\textit{STAIR}_\theta$, $\theta=2,3,4$ can be computed 
similarly. }

\begin{figure}[ht] 
  \label{ fig7} 
  \begin{minipage}[b]{0.5\linewidth}
    \centering
    \includegraphics[width=.65\linewidth]{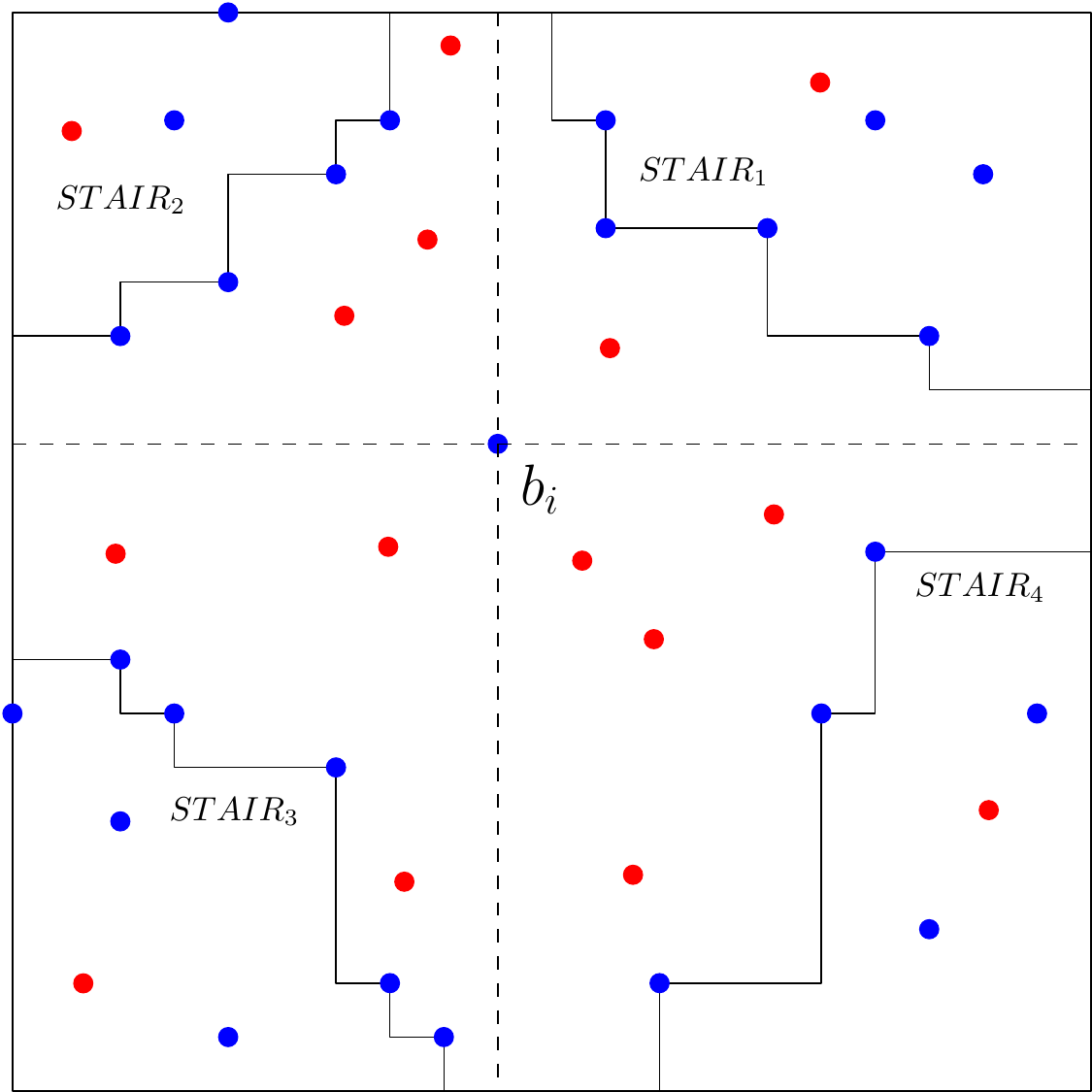} 
    \centerline{(a)}
  \end{minipage}
  \begin{minipage}[b]{0.5\linewidth}
    \centering
    \includegraphics[width=.65\linewidth]{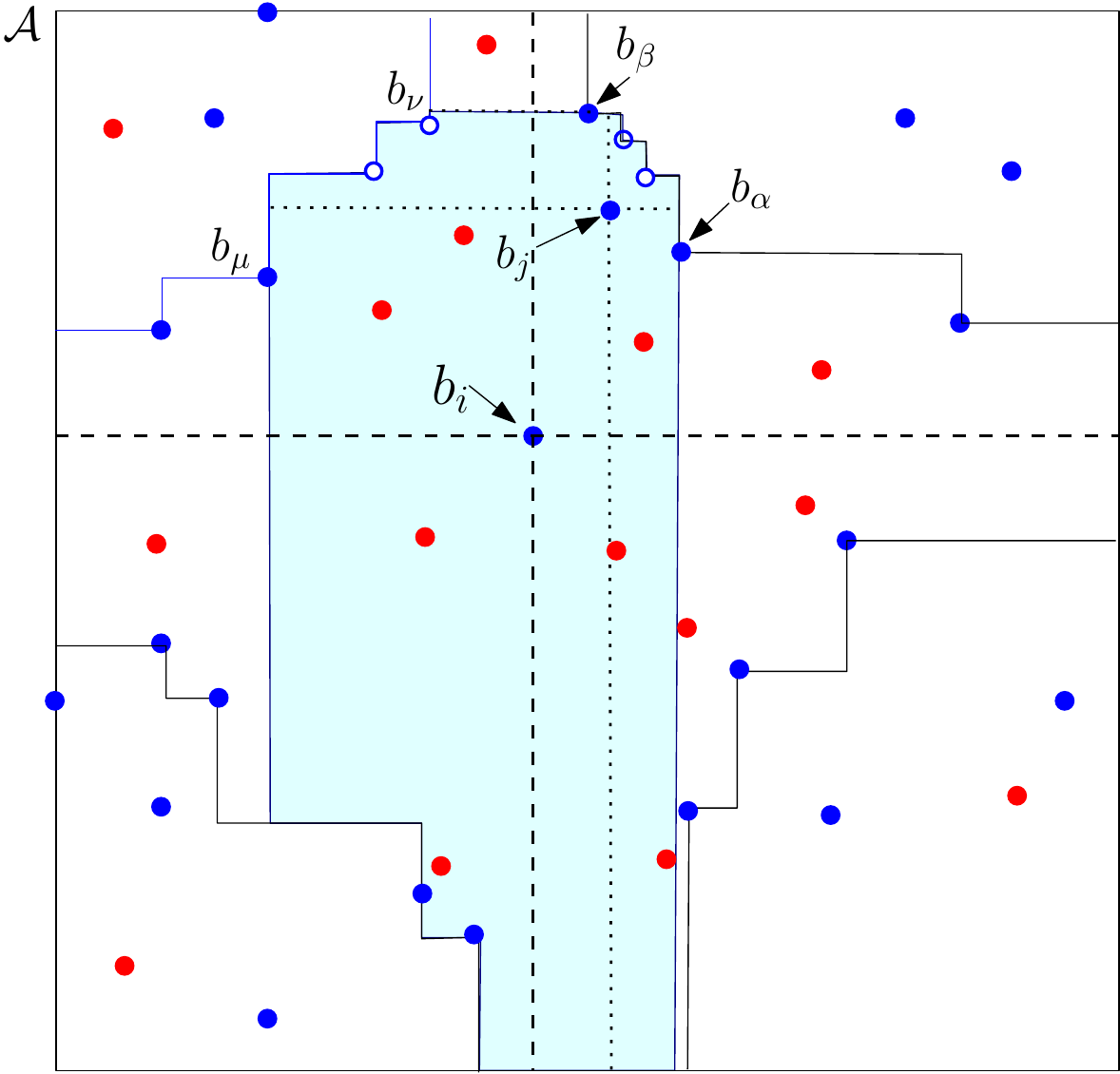} 
    \centerline{(b)}
  \end{minipage} 
  \caption{(a) Empty ortho-convex polygon around $p_i$ (b) Extracting the region in $OP$ for generating MERs with top and bottom face passing through $p_i$ and $p_j$}
  \label{3dtype2i}
\end{figure}

Thus we need to: (i) construct the in-place 2-d tree ${\cal T}_{ij}$ with the points in 
$R_{ij}$, (ii) compute all maximal empty rectangles in $OP$ that contains both 
$b_i$ and $b_j$ (see ), (iii) for each generated maximal empty rectangle ($MER$) 
perform the rectangular range counting query in ${\cal T}_{ij}$, and (iv) 
update $OP$ by inserting $b_j$ in the corresponding $\textit{STAIR}$ for
processing the next blue point $p_{j+1} \in P_b$ during this phase. The tasks (i) and (iii) 
performed as mentioned in Sections \ref{preprocessing2d} and \ref{counting2d} 
respectively. Task (ii) is explained in Section \ref{compMER3} 
(also see Algorithm  
\ref{TYPE-3}). Task (iv) is explained in Section
\ref{OPx} (also see Algorithm \ref{updatestair3d}). 

\begin{algorithm}
 \SetKwData{mini}{minimum}
 
\KwIn{The array $P_r$ and $P_b$} 
\KwOut{TYPE-3 $LRC$ of maximum size}
Sort the points in $P_b$ in decreasing order of their $z$-coordinates\;
\For (\tcc*[f]{Compute MER$(p_i)$}){$i \leftarrow 1$ \KwTo $m$}{ 
Partition the points in $P_b[i+1,i+2,\ldots,m]$ into  $P_b^\theta$, 
$\theta=1,2,3,4$\;  
$P_b^\theta, \theta \in \{1,2,3,4\}$ are sorted in decreasing order of their $z$-coordinates\;
$m_1$,$m_2$,$m_3$,$m_4$: index of the last point in each of $P_b^\theta, \theta \in \{1,2,3,4\}$ respectively\;
$\nu_1$,$\nu_2$,$\nu_3$,$\nu_4$: index of the last point in each of $STAIR_\theta, \theta \in \{1,2,3,4\}$ respectively\;
$\chi_1,\chi_2,\chi_3,\chi_4$: variables to indicate the next sweep line in $P_b^\theta, \theta \in \{1,2,3,4\}$ respectively\;
$\chi_1,\chi_2,\chi_3,\chi_4$ initialized with $1,m_1+1,m_2+1,m_3+1$ respectively\;
$count=i$\;
\While{$count\neq m$}
{
$count=count+1$\;
$z$=\mini$\{z(P_b[\chi_1]), z(P_b[\chi_2]),z(P_b[\chi_3]),z(P_b[\chi_4])\}$\;
Let, minimum attains for $P_b[\chi_\theta]$ and in quadrant $\theta$\;
Compute\_MAX\_MER($i,\chi_\theta,\theta,R_{max}$)
\tcc*[r]{call Algorithm \ref{MER2d}.}
\If {$|R_{max}| > size_{max}$} {$size_{max}=|R_{max}|$; $C=R_{max}$;}
$\USt_\theta(\chi_\theta)$\;
$\chi_\theta=\chi_\theta+1$\;
}
}
\caption{TYPE-3\_LRC($size_{max},C$)}
\label{TYPE-3}
\end{algorithm}

\subsubsection{Computation of $MER(p_i,p_j)$} \label{compMER3}

Without loss of generality, assume that $b_j$ (projection of $p_j$ on the plane $H(p_i)$) 
is in the first quadrant. If $b_j$ is in some other quadrant, then the situation is similarly 
tackled. 

If there exist any point in the $STAIR_1$ which dominates $b_j$, i.e., if there exist any 
blue point $p$ in $STAIR_1$ such that $x(p)< x(b_j)$ and $y(p)<y(b_j)$, then no axis-parallel $cLRC$ is possible whose top boundary passes through $p_i$ and bottom
boundary passes through $p_j$. Therefore we assume that $b_j$ is not dominated by any point in $STAIR_1$. We now determine
the subset of points in $\textit{STAIR}_1 \cup \textit{STAIR}_2$ that can appear in the north boundary of an $MER$ containing
both $b_i$ and $b_j$.

\begin{figure}[ht] 
  \begin{minipage}[b]{0.5\linewidth}
    \centering
    \includegraphics[width=.8\linewidth]{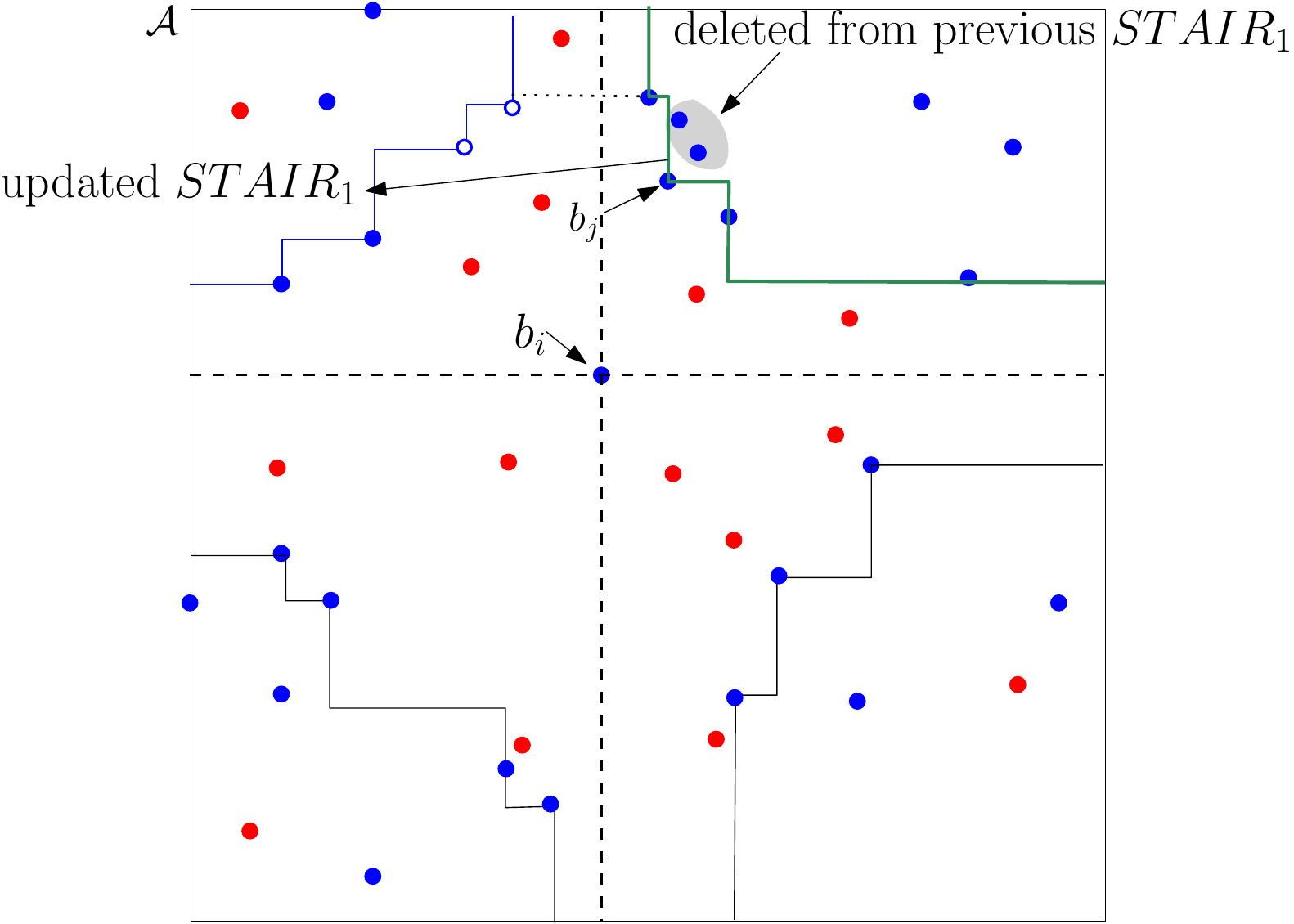} 
    \centerline{(a)}
  \end{minipage}
  \begin{minipage}[b]{0.5\linewidth}
    \centering
    \includegraphics[width=.65\linewidth]{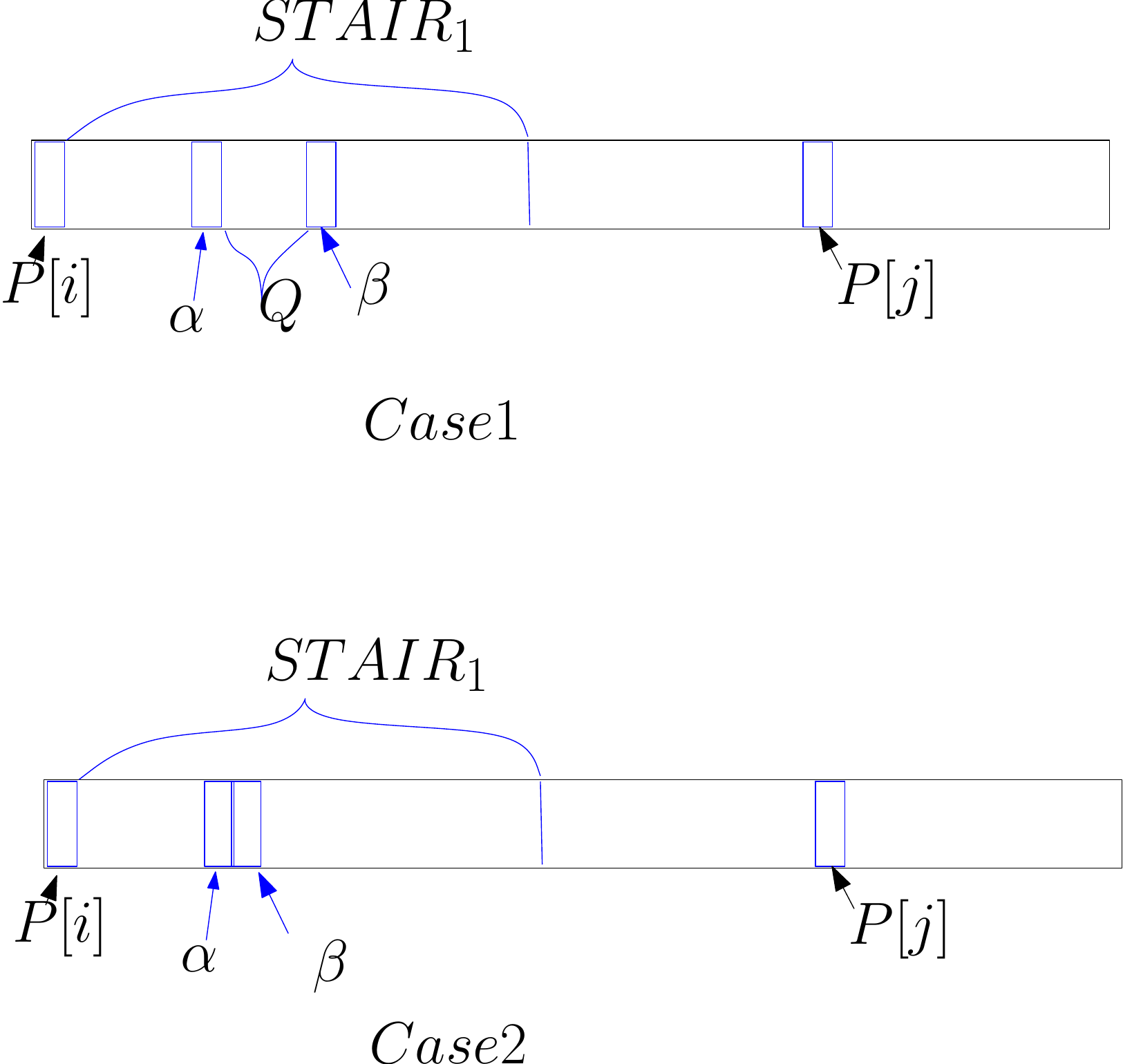} 
    \centerline{(b)}
  \end{minipage} 
  \caption{(a) Update $STAIR_1$ after processing $b_j$ w.r.t. $b_j$  and (b) corresponding array update}
  \label{figtype3}
\end{figure}

Let $\textit{STAIR}_1 = \{ b_k,k=1,2 \ldots,\nu_1 \} \subseteq B_i^1$. Let $b_\alpha \in STAIR_1$ be such that $y(b_\alpha) = \max\{b_k\in STAIR_1|y(b_k) < y(b_j)\}$ (i.e., the $y$-coordinate
of $b_\alpha$ is maximum among all the points in $STAIR_1$ whose $y$-coordinate is lesser than the $y$-coordinate of $b_j$). 
Similarly, let $b_\beta \in STAIR_1$ be such that $y(b_\beta) = 
\min\{b_k \in STAIR_1|x(b_k) < x(b_j)\}$ (i.e., $y$-coordinate of $b_\beta$ is minimum 
among all the points in $STAIR_1$ whose $x$-coordinate is lesser than the
$x$-coordinate of $p_j$). 
We define  $Q=\{b_{\alpha+1}, b_{\alpha+2}, \ldots, b_{\beta-1}\}$ $=\{b_k \in 
\textit{STAIR}_1|x(b_k) > x(b_j)~\text{and}~ y(b_k) > y(b_j)\}$  (see
Figure~\ref{3dtype2i}(b)).

All the axis-parallel $MER$s in $OP$ with north boundary passing through $b_k$, $k
\in \{ \alpha+1,\alpha+2,\ldots,\beta \}$ and containing $p_i$ in its proper interior will contain $b_j$ also.
We draw the projections of $b_j$ and $b_\beta$ on $\textit{STAIR}_2$. Let these two points be $\mu$ and $\nu$,
respectively. If $x(\mu) = x(\nu)$, then no point on $\textit{STAIR}_2$ can appear on the north boundary of a
desired  axis-parallel $MER$. But if $x(\mu) < x(\nu)$, then all the points $p \in \textit{STAIR}_2$ satisfying
$x(\mu) < x(p) < x(\nu)$ can appear on the north boundary of a desired axis-parallel $MER$. In Figure \ref{3dtype2i}(b),
the set of points that can appear on the north boundary of an $MER$ are marked with empty dots.  The method of
computing an axis-parallel $MER$ with a point $p \in \textit{STAIR}_1 \cup \textit{STAIR}_2$ on its north boundary
is given in Algorithm \ref{MER2d}.

\begin{algorithm}
\SetKwInOut{Kw}{Work-Area} \SetKwData{index}{index-of} \SetKwData{maxi}{maximum}\SetKwData{mini}{minimum}
\SetKwData{size}{size}
\small{
\KwIn{$STAIR_1$ = $B[1,2,\ldots,\nu_1]$, $STAIR_2$ = $B[m_1+1,m_1+2,\ldots,\nu_2]$, 
$STAIR_3$ = $B[m_2+1,m_2+2,\ldots,\nu_3]$, $STAIR_4$ = $B[m_3+1,m_3+2,\ldots,\nu_4]$, where the array $B=P_b[i+1, \ldots m]$; $m_\theta$ = number of points 
of $B$ in $\theta$-th quadrant\; }

\Kw{$M$: location to compute the size of the axis-parallel $MER$ containing $p_i,p_j$;
$R$: stores the $(north,south,east,west)$ sides of a rectangle \tcc*{$b_i$: projection of $p_i$} }
 \KwOut{$R_{max}$ \tcc*[r]{red rectangle containing maximum red points} }

$MAX\_size = 0$\;
$\alpha$ = \index {\mini{$y(B[k])$}}: $\forall ~k \in \{m_1+1,\ldots,\nu_2\}$ and $y(B[k]))>y(B[j])$\;
$\beta$ = \index {\maxi{$y(B[k])$}}: $\forall ~k \in \{m_1+1,\ldots,\nu_2\}$ and $x(B[k]))>x(B[j])$; $\beta=\beta+1$\;
$\mu$ = \index {\mini{$y(B[k])$}}: $\forall ~k \in \{m_2+1,\ldots,\nu_3\}$ and $y(B[k]))>y(B[j])$\;
$\nu$ = \index {\maxi{$y(B[k])$}}: $\forall ~k \in \{m_2+1,\ldots,\nu_3\}$ and $y(B[k]))<y(B[\beta])$\;

\For(\tcc*[f]{Call MER with the feasible points of $STAIR_1$ as top boundary  }){$k \leftarrow \alpha$ \KwTo $\beta$}{
$north=P[k]$; $east=P[k-1]$\;
$\theta=$ \index {\maxi{$y(B[\ell])$}}: $\forall ~\ell \in \{m_1+1,\ldots,\nu_2\}$ and $y(B[\ell]))<y(B[k]))$\;
$\psi=$ \index {\maxi{$x(B[\ell])$}}: $\forall ~\ell \in \{m_3+1,\ldots,\nu_4\}$ and $x(B[\ell]))<x(B[k]))$\;
$\phi=$ \index {\maxi{$y(B[\ell])$}}: $\forall ~\ell \in \{m_2+1,\ldots,\nu_3\}$ and $x(B[\ell]))<x(B[\theta])$\;
$\phi'=$ \index {\mini{$y(B[\ell])$}}: $\forall ~\ell \in \{m_2+1,\ldots,\nu_3\}$ and $y(B[\ell]))>y(B[\psi])$\;
\If(\tcc*[f] {Only one MER is possible}){$\phi'>\phi$} {$west =P[\theta]$;
$south=P[\psi]$; $R = (north,east, south,west)$\;
\size=\QIT(${\cal T}_i,R$)\;
	  {\bf if} {\size $>$ MAX\_size} {\bf then} {MAX\_size $\leftarrow$ \size; $R_{max} \leftarrow R$\;} }	  	  	  
\If(\tcc*[f] {Multiple $(\geq 2)$ MERs are possible}){$\phi'\leq \phi$} {
$south=P[\psi']$\; 
\For{$\ell=\psi'$ \KwTo $\psi$}{
$west=P[\ell]$; $R=(north,east,south,west)$; $south=P[\ell]$\;
\size=\QIT(${\cal T}_i,R$)\;
	  {\bf if} {\size $>$ MAX\_size} {\bf then} {MAX\_size $\leftarrow$ \size; $R_{max} \leftarrow R$\;} }
$west=P[\alpha]$;	 $R=(north,east,south,west)$\; 
 \size=\QIT(${\cal T}_i,R$)\;
	  {\bf if} {\size $>$ MAX\_size} {\bf then} {MAX\_size $\leftarrow$ \size; $R_{max} \leftarrow R$\;} }
}
\For(\tcc*[f]{Call MER with the feasible points of $STAIR_2$ as top boundary  }){$k \leftarrow \mu$ \KwTo $\nu$}{	 
 $west=P[k-1]$; $top=P[k]$; $east=P[\theta']$\;
$\theta_1$ = \index{\maxi{$y(B[\ell])$}}:$\forall ~\ell \in \{\theta \ldots, \nu_1\}$ and $y(B[\ell]) < y(B[\mu])$\;
$\psi_1=$ \index {\maxi{$x(B[\ell])$}}: $\forall ~\ell \in \{m_2+1,\ldots,\nu_3\}$ and $x(B[\ell]))>x(B[k-1]))$\;
$\phi_1=$ \index {\maxi{$y(B[\ell])$}}: $\forall ~\ell \in \{m_3+1,\ldots,\nu_4\}$ and $x(B[\ell]))<x(B[\theta_1])$\;
$\phi_2=$ \index {\mini{$y(B[\ell])$}}: $\forall ~\ell \in \{m_2+1,\ldots,\nu_3\}$ and $y(B[\ell]))>y(B[\psi_1])$\;
\If(\tcc*[f] {Only one MER is possible}){$\phi_2>\phi_1$} {$east =P[\theta_1]$;
$south=P[\psi_1]$; $R = (north,east, south,west)$\;
\size=\QIT(${\cal T}_i,R$)\;
	  {\bf if} {\size $>$ MAX\_size} {\bf then} {MAX\_size $\leftarrow$ \size; $R_{max} \leftarrow R$\;} }
\If(\tcc*[f] {Multiple MER is possible}){$\phi_2\leq \phi_1$}{
$south=P[\phi_2]$\;
\For{$\ell=\phi_2$ \KwTo $\phi_1$}{
$west=P[\ell]$; $R=(north,east,south,west)$; $south=P[\ell]$\;
\size=\QIT(${\cal T}_i,R$)\;
	  {\bf if} {\size $>$ MAX\_size} {\bf then} {MAX\_size $\leftarrow$ \size; $R_{max} \leftarrow R$\;} }
$west=P[\alpha_1]$;	 $R=(north,east,south,west)$\; 
 \size=\QIT(${\cal T}_i,R$)\;
	  {\bf if} {\size $>$ MAX\_size} {\bf then} {MAX\_size $\leftarrow$ \size; $R_{max} \leftarrow R$\;} }
	  }	  
}	
 
\caption{Compute\_MAX\_MER($i,j,\theta,R_{max}$)}
\label{MER2d}
\end{algorithm}

\subsubsection{Updating $OP$} \label{OPx}
 
After computing the set of axis-parallel $MER$s in $OP$ containing both the projected points $b_i$ and $b_j$ in its interior,
instead of recomputing the whole ortho-convex polygon again to process the next point $p_{j+1} \in P_b$, we update $OP$ by inserting
$b_j$ in the respective $STAIR$ (see Figure \ref{figtype3}(a)).

Without loss of generality, assume that $b_j$ lies in the first quadrant. After inserting $b_j$ in $STAIR_1$, none of the
points in $Q \in STAIR_1$ will participate in forming $MER$ while processing points $p_k \in P_b$ with $z(p_k)<z(p_j)$. So,
we need to remove the members in $Q$ from $\textit{STAIR}_1$.
This can be done by using the algorithm for stable sorting \cite{KP92}, where the elements in $Q$ will assume the value 1 of the
given (0, 1)-valued selection function $f$, and will stably move to the end of
$\textit{STAIR}_1$. A simple procedure for this task is given in \cite{BMM07} in
the context of stably selecting a sorted subset. We tailored that procedure for
our purpose as follows:
 
We maintain two index variables $\alpha$ and $\beta$;
$\alpha +1$ and $\beta -1$ indicates the starting  and ending positions of the $Q$, respectively. Now, two cases may arise depending
on whether $|Q|=0$ or not. 
\begin{description}
 \item[$|Q| \neq 0$]: See Case 1 of Figure  \ref{figtype3}(b). Here, we  need to remove $Q$ from $STAIR_1$ and
 appropriately insert $p_j$ into the stair. We do this by the following way:
\begin{itemize}
 \item First, by swapping $b_j$ and $b_{\alpha +1}$, we insert $b_j$ in the proper position.
\item Now, we need to move out $b_{\alpha+1}, \ldots, b_{\beta-1}$ from the $STAIR_1$. This can be done by a sequence of swap
operations:  swap($P[r], P[r-(\beta-\alpha-2)]$, starting from $r=\beta$ until $r=\nu_1$, where $\nu_1$ denotes the end of $STAIR_1$.
\item Finally, we set $\nu_1$ as $\nu_1-(\beta-\alpha-2)$.
\end{itemize}

 \item[$|Q| = 0$]: See Case 2 of Figure  \ref{figtype3}(b). Here, we need only to insert $b_j$ into the stair. We do
 this by first swapping $(P[\nu_1+1],P[j])$ and then a sequence of swapping ($P[r], P[r+1]$) starting from $r=\nu_1$  until $r=\beta$.
 Finally, we set $\nu_1$ as $\nu_1+1$.
\end{description}

Clearly, this updating $OP$ needs   $O(|P_i^1|)$ time in the worst case. 

\begin{algorithm}
\SetKwData{index}{index-of} \SetKwData{maxi}{maximum}\SetKwData{mini}{minimumorangec}
 \KwIn{$STAIR_1$ corresponding to $p_i$, the projection $b_j$}
 \KwOut{updated $STAIR_1$}
$\alpha=$ \index {\maxi{$y(b_k$}}: $\forall ~k \in STAIR_1$ and $y(b_k)<y(b_j)$\;

$\beta=$ \index {\mini{$y(b_k$}}: $\forall ~k \in STAIR_1$ and $x(b_k)<x(b_j)$\;
\If{$(\beta- \alpha) > 1 $}
{swap($P[j], P[\alpha +1]$)\;
$k= \beta - \alpha-2$\;
\For{$r\leftarrow \beta$  \KwTo  $\nu_1$}
{swap($P[r], P[r-k]$)\;
}
$\nu_1= \nu_1 -k$\;
}
\Else{
swap($P[\nu_1+1],P[j]$)\;
\For{$r \leftarrow \nu_1$ \KwTo $\beta$}
{swap($P[r], P[r+1]$)\;
}
$\nu_1= \nu_1 +1$\;
}
\caption{ $\USt_1$($j$)}
\label{updatestair3d}
\end{algorithm}

After computing the largest {\it type-3} axis-parallel $LRC$ with $p_i$ on its top boundary, we need to sort the points again with
respect to their $z$-coordinates for the processing of $p_{i+1}$. 

Thus we have the following result:

\begin{lemma} \label{res-type3}
The time required for processing $p_i$ is $O(m^2+C_i'\sqrt{n} +mn\log n)$ in the worst case, where $C_i'$ is the number of
{\em type-3} axis-parallel $LRC$s with $p_i$ on its top boundary. 
\end{lemma}

\begin{proof}
The worst case time required for computing $MER(p_i,p_j)$ is $O(|P_{ij}|+C_{ij})$, where $P_{ij}$ denotes the number of points
inside the horizontal slab bounded by $H(p_i)$ and $H(p_j)$, and $C_{ij}$ denotes the number of axis-parallel $MER$s
containing both $b_i$ and $b_j$ inside $OP$ with the projection of points $B_{ij}$ on $H(p_i)$. In order to compute
the largest {\em type-3} axis-parallel $LRC$ with $p_i$ on its top boundary, we need to compute $MER(b_i,b_j)$ for all
$j>i$, $C_i'=\displaystyle\sum_{j=i+1}^n C_{ij}$, and $\displaystyle\sum_{j=i+1}^n |P_{ij}| =O((m-i)^2)$. For each of the cuboid $C_i'$, the
in-place counting query in the corresponding ${\cal T}_{ij}$ requires $\sqrt{n}$ time (using Lemma \ref{count}). The
last part of the time complexity follows due to the fact that, for every point $p_j, j=i+1,\ldots,m$ we need to construct the in-place
$2d$-tree. Also after the processing of each $p_i\in P_b$, the sorting step takes $O(n\log n)$ time. 
\end{proof}

Lemma \ref{res-type1}, \ref{res-type2} and \ref{res-type3} lead to the following result. 

\begin{theorem}
The worst case time complexity of our in-place algorithm for computing the axis-parallel largest monochromatic cuboid ($LMC$)
is $O(m^3\sqrt{n}+m^2n\log n)$, and it takes $O(1)$ extra space.
\end{theorem}

\subsection*{Acknowledgment:} The authors acknowledge the valuable constructive suggestions given by the reviewer regarding the presentation of the paper.

\end{document}